\documentclass[12pt,a4paper]{amsart}

\selectfont

\usepackage{qtree}
\usepackage{forest}
\usepackage{mathtools} 

\usepackage{caption}

\usepackage{etoolbox}
\usepackage{epic,eepic,ecltree}
\usepackage{graphicx}
\usepackage{tikz}
\usepackage{tikz-cd}
\usetikzlibrary{decorations.text}
\usepackage{amssymb,color, euscript, enumerate}
\usepackage{amsthm}
\usepackage{amsmath}
\usepackage{braket}
\usepackage{amscd}
\usepackage{txfonts}
\usepackage{comment}
\usepackage{bm}
\usepackage{amscd}
\numberwithin{equation}{section}

\makeatletter
\let\old@tocline\@tocline
\let\section@tocline\@tocline
\newcommand{\subsection@dotsep}{4.5}
\newcommand{\subsubsection@dotsep}{4.5}
\patchcmd{\@tocline}
  {\hfil}
  {\nobreak
     \leaders\hbox{$\m@th
        \mkern \subsection@dotsep mu\hbox{.}\mkern \subsection@dotsep mu$}\hfill
     \nobreak}{}{}
\let\subsection@tocline\@tocline
\let\@tocline\old@tocline

\patchcmd{\@tocline}
  {\hfil}
  {\nobreak
     \leaders\hbox{$\m@th
        \mkern \subsubsection@dotsep mu\hbox{.}\mkern \subsubsection@dotsep mu$}\hfill
     \nobreak}{}{}
\let\subsubsection@tocline\@tocline
\let\@tocline\old@tocline

\let\old@l@subsection\l@subsection
\let\old@l@subsubsection\l@subsubsection

\def\@tocwriteb#1#2#3{%
  \begingroup
    \@xp\def\csname #2@tocline\endcsname##1##2##3##4##5##6{%
      \ifnum##1>\c@tocdepth
      \else \sbox\z@{##5\let\indentlabel\@tochangmeasure##6}\fi}%
    \csname l@#2\endcsname{#1{\csname#2name\endcsname}{\@secnumber}{}}%
  \endgroup
  \addcontentsline{toc}{#2}%
    {\protect#1{\csname#2name\endcsname}{\@secnumber}{#3}}}%

\newlength{\@tocsectionindent}
\newlength{\@tocsubsectionindent}
\newlength{\@tocsubsubsectionindent}
\newlength{\@tocsectionnumwidth}
\newlength{\@tocsubsectionnumwidth}
\newlength{\@tocsubsubsectionnumwidth}
\newcommand{\settocsectionnumwidth}[1]{\setlength{\@tocsectionnumwidth}{#1}}
\newcommand{\settocsubsectionnumwidth}[1]{\setlength{\@tocsubsectionnumwidth}{#1}}
\newcommand{\settocsubsubsectionnumwidth}[1]{\setlength{\@tocsubsubsectionnumwidth}{#1}}
\newcommand{\settocsectionindent}[1]{\setlength{\@tocsectionindent}{#1}}
\newcommand{\settocsubsectionindent}[1]{\setlength{\@tocsubsectionindent}{#1}}
\newcommand{\settocsubsubsectionindent}[1]{\setlength{\@tocsubsubsectionindent}{#1}}

\renewcommand{\l@section}{\section@tocline{1}{\@tocsectionvskip}{\@tocsectionindent}{}{\@tocsectionformat}}%
\renewcommand{\l@subsection}{\subsection@tocline{2}{\@tocsubsectionvskip}{\@tocsubsectionindent}{}{\@tocsubsectionformat}}%
\renewcommand{\l@subsubsection}{\subsubsection@tocline{3}{\@tocsubsubsectionvskip}{\@tocsubsubsectionindent}{}{\@tocsubsubsectionformat}}%
\newcommand{\@tocsectionformat}{}
\newcommand{\@tocsubsectionformat}{}
\newcommand{\@tocsubsubsectionformat}{}
\expandafter\def\csname toc@1format\endcsname{\@tocsectionformat}
\expandafter\def\csname toc@2format\endcsname{\@tocsubsectionformat}
\expandafter\def\csname toc@3format\endcsname{\@tocsubsubsectionformat}
\newcommand{\settocsectionformat}[1]{\renewcommand{\@tocsectionformat}{#1}}
\newcommand{\settocsubsectionformat}[1]{\renewcommand{\@tocsubsectionformat}{#1}}
\newcommand{\settocsubsubsectionformat}[1]{\renewcommand{\@tocsubsubsectionformat}{#1}}
\newlength{\@tocsectionvskip}
\newcommand{\settocsectionvskip}[1]{\setlength{\@tocsectionvskip}{#1}}
\newlength{\@tocsubsectionvskip}
\newcommand{\settocsubsectionvskip}[1]{\setlength{\@tocsubsectionvskip}{#1}}
\newlength{\@tocsubsubsectionvskip}
\newcommand{\settocsubsubsectionvskip}[1]{\setlength{\@tocsubsubsectionvskip}{#1}}

\patchcmd{\tocsection}{\indentlabel}{\makebox[\@tocsectionnumwidth][l]}{}{}
\patchcmd{\tocsubsection}{\indentlabel}{\makebox[\@tocsubsectionnumwidth][l]}{}{}
\patchcmd{\tocsubsubsection}{\indentlabel}{\makebox[\@tocsubsubsectionnumwidth][l]}{}{}

\newcommand{\@sectypepnumformat}{}
\renewcommand{\contentsline}[1]{%
  \expandafter\let\expandafter\@sectypepnumformat\csname @toc#1pnumformat\endcsname%
  \csname l@#1\endcsname}
\newcommand{\@tocsectionpnumformat}{}
\newcommand{\@tocsubsectionpnumformat}{}
\newcommand{\@tocsubsubsectionpnumformat}{}
\newcommand{\setsectionpnumformat}[1]{\renewcommand{\@tocsectionpnumformat}{#1}}
\newcommand{\setsubsectionpnumformat}[1]{\renewcommand{\@tocsubsectionpnumformat}{#1}}
\newcommand{\setsubsubsectionpnumformat}[1]{\renewcommand{\@tocsubsubsectionpnumformat}{#1}}
\renewcommand{\@tocpagenum}[1]{%
  \hfill {\mdseries\@sectypepnumformat #1}}

\let\oldappendix\appendix
\renewcommand{\appendix}{%
  \leavevmode\oldappendix%
  \addtocontents{toc}{%
    \protect\settowidth{\protect\@tocsectionnumwidth}{\protect\@tocsectionformat\sectionname\space}%
    \protect\addtolength{\protect\@tocsectionnumwidth}{2em}}%
}
\makeatother

\makeatletter
\settocsectionnumwidth{2em}
\settocsubsectionnumwidth{2.5em}
\settocsubsubsectionnumwidth{3em}
\settocsectionindent{1pc}%
\settocsubsectionindent{\dimexpr\@tocsectionindent+\@tocsectionnumwidth}%
\settocsubsubsectionindent{\dimexpr\@tocsubsectionindent+\@tocsubsectionnumwidth}%
\makeatother

\settocsectionvskip{10pt}
\settocsubsectionvskip{0pt}
\settocsubsubsectionvskip{0pt}
    
\settocsectionformat{\bfseries}
\settocsubsectionformat{\mdseries}
\settocsubsubsectionformat{\mdseries}
\setsectionpnumformat{\bfseries}
\setsubsectionpnumformat{\mdseries}
\setsubsubsectionpnumformat{\mdseries}

\let\oldtableofcontents\tableofcontents
\renewcommand{\tableofcontents}{%
  \vspace*{-\linespacing}
  \oldtableofcontents}

\newcommand{\cS}{\mathcal{S}}

\newcommand{\Harm}{\mathrm{Harm}}

\newcommand{\T}{\mathfrak{T}}
\newcommand{\tT}{\tilde{T}}
\newcommand{\Cc}{C_c^\infty}

\newcommand{\Emb}{\mathrm{Mfld}_{d,\mathrm{emb}}^{\mathrm{CO}}}
\newcommand{\Embc}{\mathrm{Mfld}_{d}^{\mathrm{CO}}}

\newcommand{\EmbR}{\mathrm{Mfld}_{d}^{\mathrm{Riem}}}
\newcommand{\cl}{\mathrm{cl}}
\newcommand{\Disk}{\mathbb{CE}_d^{\mathrm{emb}}}
\newcommand{\Diskt}{\mathbb{CE}_2^{\mathrm{emb}}}
\newcommand{\SO}{\mathrm{SO}}

\newcommand{\HY}{\mathrm{F}_{\mathrm{CL}}}

\newcommand{\tDisk}{{\mathrm{Disk}}_{d,\mathrm{emb}}^{\mathrm{CO}}}

\newcommand{\vol}{{\mathrm{vol}}}

\newcommand{\Hom}{{\mathrm{Hom}}}

\newcommand{\cC}{{\mathcal C}}

\newcommand{\N}{\mathbb{N}}
\newcommand{\Z}{\mathbb{Z}}
\newcommand{\R}{\mathbb{R}}
\newcommand{\C}{\mathbb{C}}

\newcommand{\pr}{{\mathrm{pr}}}

\newcommand{\Om}{\Omega}

\newcommand{\Conf}{\mathrm{Conf}}

\newcommand{\sod}{\mathrm{so}(d+1,1)}

\newcommand{\alg}{\text{alg}}
\newcommand{\Span}{\mathrm{Span}}

\newcommand{\Hpt}{\mathcal{J}^\mathrm{pt}(\bD)}
\newcommand{\Hp}{\mathcal{J}^\mathrm{pt}}
\newcommand{\OD}{\mathcal{O}(\bD)}
\newcommand{\ODm}{\mathcal{O}(\bD,*0)}

\newcommand{\geom}{\mathrm{geom}}

\newcommand{\Vh}{\mathfrak X (\bD)}
\newcommand{\Vm}{\mathfrak X (\bD,*0)}

\newcommand{\Res}{\mathrm{Res}}
\newcommand{\bmJ}{J[\bm{m};\bm{a}]}
\newcommand{\blJ}{J^{(\bm l)}({\bm b})}

\newcommand{\Pair}{\mathrm{Pair}}

\newcommand{\std}{{\text{std}}}

\newcommand{\va}{\bm{1}}

\newcommand{\id}{{\mathrm{id}}}

\newcommand{\tW}{{\tilde{W}}}
\newcommand{\z}{{\bar{z}}}

\newcommand{\w}{{\bar{w}}}

\newcommand{\m}{{\rho}}

\newcommand{\pa}{{\partial}}

\newcommand{\Vect}{{\underline{\text{Vect}}_\R}}

\newcommand{\cN}{{\mathcal{N}}}

\newcommand{\al}{\alpha}
\newcommand{\ep}{\epsilon}
\newcommand{\be}{\beta}

\newcommand{\de}{\delta}
\newcommand{\D}{\mathbb{D}}
\newcommand{\om}{\omega}
\newcommand{\la}{\lambda}

\newcommand{\si}{\sigma}

\newcommand{\g}{{\mathfrak{g}}}

\newcommand{\ft}{\frac{1}{2}}

\newcommand{\Sym}{\mathrm{Sym}}

\newcommand{\Hf}{{\mathrm{H}}}

\newcommand{\CE}{\mathbb{CE}_d^{\mathrm{emb}}}
\newcommand{\CEt}{\mathbb{CE}_2^{\mathrm{emb}}}

\newcommand{\bD}{\mathbb{D}}
\newcommand{\oD}{{\overline{\mathbb{D}}}}

\newcommand{\supp}{\mathrm{supp}}

\newcommand{\norm}[1]{\left\lVert #1 \right\rVert}

\newcommand{\CFT}{{\mathrm{CFT}}}

\newtheorem{thm}{Theorem}[section]
\newtheorem{dfn}[thm]{Definition}
\newtheorem{lem}[thm]{Lemma}
\newtheorem{prop}[thm]{Proposition}
\newtheorem{cor}[thm]{Corollary}
\newtheorem{rem}[thm]{Remark}

\theoremstyle{remark}

\begin{document}

\begin{center}
{{\LARGE \bf Prefactorization algebras for the conformal Laplacian: Central charge and Hilbert Fock space}
} \par \bigskip

\renewcommand*{\thefootnote}{\fnsymbol{footnote}}
{\normalsize
Yuto Moriwaki \footnote{email: \texttt{moriwaki.yuto (at) gmail.com}}
}
\par \bigskip
{\footnotesize Interdisciplinary Theoretical and Mathematical Science Program (iTHEMS)\\
Wako, Saitama 351-0198, Japan}

\par \bigskip
\end{center}

\noindent

\begin{center}
\textbf{\large Abstract}
\end{center}

Let $d \geq 2$. We consider the symmetric monoidal category of oriented Riemannian $d$-manifolds with conformal open embeddings. The prefactorization algebra associated with the conformal Laplacian defines a symmetric monoidal functor from this category to real vector spaces.
For Euclidean domains $U\subset\mathbb{R}^d$, the value of this functor is identified, via the Green function, with the symmetric algebra on the topological dual of the space of harmonic functions. For $d \geq 3$ this identification is natural under all conformal transformations, while in dimension two, its failure of naturality is governed by a harmonic cocycle, which gives rise to the Virasoro central charge.
For the unit disk, the resulting vector space carries an algebra structure over the operad of conformal disk embeddings and admits a canonical dense embedding into the Hilbert Fock space. In dimension two, this statement holds after restricting to a codimension-one subspace, as suggested by logarithmic CFT.

\vspace{3mm}

\tableofcontents

\begin{center}
\textbf{\large Introduction}
\end{center}

The factorization algebras developed by Costello-Gwilliam provide a formulation of quantum field theory \cite{CG,CG2}.
For the quantum field theory known as the free scalar field, the associated factorization algebra assigns to a Riemannian manifold $(M,g)$ a chain complex, called the BV complex, whose homology yields a symmetric monoidal functor from the symmetric monoidal category $\EmbR$ of $d$-dimensional Riemannian manifolds and isometric open embeddings to the category of (differentiable) vector spaces \cite[Chapter~6.3]{CG}.
In this paper, we study its analogue in conformal Riemannian geometry by considering the prefactorization algebra associated with the {\it conformal Laplacian}.

Fix $d \geq 2$. Let $\Emb$ denote the symmetric monoidal category whose objects are $d$-dimensional oriented Riemannian manifolds without boundary and whose morphisms are orientation-preserving conformal open embeddings.
On a Riemannian manifold $(M,g)$, the conformal Laplacian (also known as the Yamabe operator) \cite{Yamabe}
\begin{align*}
L_g = \Delta_g + \frac{d-2}{4(d-1)}S_g
\end{align*}
transforms under a change of metric $\hat{g} = e^{2\om} g$ ($\om:M \rightarrow \R$ a smooth function) as
\begin{align*}
L_{\hat{g}}(f) = e^{-\frac{d+2}{2}\om} L_g (e^{\frac{d-2}{2}\om} f)\quad\quad\quad(f \in C^\infty(M)).
\end{align*}
Using this conformal covariance, we show that the prefactorization algebra associated with the conformal Laplacian defines a symmetric monoidal functor (Theorem~\ref{prop_functor})
\[
\HY: \Emb \rightarrow \Vect.
\]

Following Costello--Gwilliam, the vector space $\HY(M,g)$ can be computed, whenever a Green's function $G(x,y)$ of the conformal Laplacian $L_g$ exists, via a linear isomorphism
\begin{align}
\Psi_G:\HY(M,g) \overset{\cong}{\rightarrow} \mathrm{Sym}(\mathrm{cok}(L_g:\Cc(M) \rightarrow \Cc(M))).
\label{eq_intro_psi}
\end{align}
Here $\Cc(M)$ denotes the space of real valued smooth functions with compact support.
On non-compact manifolds a Green's function (i.e.\ a fundamental solution of $L_g$) typically depends on \emph{boundary conditions}.
In particular, for $\Psi_G$ to define a natural transformation with respect to conformal open embeddings, one needs a fundamental solution satisfying conformally invariant boundary conditions.
In two dimensions, no such solution exists, and the failure of naturality of $\Psi_G$ is measured by a {harmonic cocycle}.

%
%
The main results of this paper are as follows:
\begin{enumerate}
\item
For an open subset $U\subset \R^d$ of the flat Riemannian manifold $(\R^d,g_{\std})$, we show that $\HY(U)$ is isomorphic to $\Sym H'(U)$ as vector spaces, where $H'(U)$ is the topological dual of the space $H(U)$ of harmonic functions on $U$.
\item
If $d \geq 3$, then for the unit disk
$\bD_d = \{x \in \R^d \mid |x|<1\}$,
we show that $\HY(\bD_d)$ embeds densely into the Hilbert Fock space $\widehat{\mathrm{Sym}}H_\CFT$, and we describe its image explicitly.
Here $H_\CFT$ is a Hilbert space carrying a unitary representation of $\SO^+(d,1)$.
In dimension $d=2$, the analogous embedding holds after restricting to a codimension-one subspace, reflecting the logarithmic features of the massless free scalar theory.
%
\item
We explicitly describe the algebra structure on $\HY(\bD_d)$ given by the restriction of $\HY$ to the full monoidal subcategory
$\tDisk \subset \Emb$ whose objects are disjoint unions of disks $\sqcup_n \bD_d$ ($n\geq 0$).
In particular, for $d=2$, we show that the algebra structure admits a quantum correction governed by an explicit cocycle of harmonic functions $H_\phi$, which gives rise to the central charge.
\end{enumerate}

We emphasize that, as already explained by Costello--Gwilliam \cite{CG}, the isomorphism \eqref{eq_intro_psi} is not merely a computational tool for homology, but rather encodes a choice of boundary conditions in quantum field theory.
Accordingly, physical quantities such as \emph{correlation functions} are defined only after choosing \eqref{eq_intro_psi}.
In other mathematical formulations of quantum field theory, such as the G\r{a}rding--Wightman axioms and the Osterwalder--Schrader axioms, the theory is defined in terms of correlation functions \cite{GW,OS1,OS2}.
Thus, in comparing factorization algebras with these axiomatic frameworks, one must incorporate an appropriate choice of $\Psi_G$.

In this paper, for a (non-compact) domain $U \subset \R^d$ where the Green's function is not unique, we study how the choice of such boundary conditions changes under conformal open embeddings. In particular, for $d=2$, we rederive the central charge from the viewpoint of factorization algebras. Another key observation is that the value on the flat unit disk, $\HY(\bD_d) \in \Vect$, is closely related to the \emph{Hilbert Fock space} $\widehat{\mathrm{Sym}}H_\CFT$ appearing in axiomatic QFT. In particular, this observation raises the question of whether the algebra structure in (3), defined on a dense subspace, extends to a bounded operator on the Hilbert space $\widehat{\mathrm{Sym}}H_\CFT$.
It turns out that the algebra structure on $\HY(\bD_d)$ is not necessarily given by bounded operators \cite{MLeft,MBergman}.

A factorization algebra assigns a multiplication map to every isometric open embedding.
In \cite{MLeft,MBergman}, conformal open embeddings that give rise to bounded operators are characterized in geometric or analytic terms, thereby providing a refinement of factorization algebras compatible with the Hilbert space structure.
%
Our original motivation for focusing on the unit disk $\bD_d$ comes from the study of the relationship between vertex operator algebras and axiomatic or algebraic quantum field theory \cite{CKLW,AGT,AMT}; see \cite{MLeft,MBergman} for further background.
For related viewpoints on the relation between vertex operator algebras and factorization algebras, see also \cite{CG,Bru,Vic,Nis}.

Hereafter, we will describe the main results and the organization of the paper.

\vspace{3mm}

\noindent
\begin{center}
{0.1. \bf 
Category of conformal Riemannian manifolds
}
\end{center}

Let $\Emb$ be the symmetric monoidal category of $d$-dimensional Riemannian manifolds and orientation-preserving conformal open embeddings.
We denote by $\tDisk$ the full monoidal subcategory of $\Emb$ generated by the unit disk $\bD_d$,
whose objects are disjoint unions of disks $\sqcup_n \bD_d$ ($n \geq 0$).
This category is described by the operad consisting of conformal open embeddings from disjoint unions of disks into a disk
 (Fig. \ref{fig_intro_operad}):
\begin{align*}
\CE(n) = \Hom_{\Emb}(\sqcup_n \bD_d, \bD_d),\quad \quad\quad n \geq 0.
\end{align*}

\vspace{-2mm}

\begin{figure}[h]
\begin{minipage}[c]{.44\textwidth}
\centering
    \includegraphics[scale=0.3]{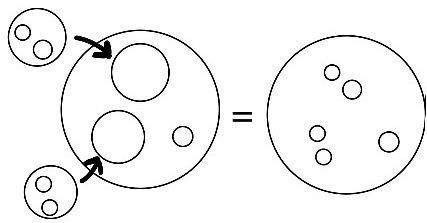}
    \caption{$\CE$-operad}\label{fig_intro_operad}
\end{minipage}
\end{figure}

This operad is the analogue, in conformal Riemannian geometry, of the $d$-disk operad in algebraic topology.
By restricting the symmetric monoidal functor
\begin{align}
\HY: \Emb \rightarrow \Vect,
\label{intro_HY}
\end{align}
defined by the prefactorization algebra associated with the conformal Laplacian, to $\tDisk \subset \Emb$, we obtain a $\CE$-algebra structure on the vector space $\HY(\bD_d)$.
As observed by Kawahira--Shigemura \cite{KS}, in the non-compact case $\HY(M,g)$ may in general be infinite-dimensional.
In this paper, 
we will give an explicit description of the $\CE$-algebra structure on $\HY(\bD_d)$.

%
%
%

\vspace{3mm}

\noindent
\begin{center}
{0.2. \bf 
Boundary conditions on flat Euclidean spaces and central charge
}
\end{center}

On $\R^d$ and its open subsets, one may take as a fundamental solution of the (conformal) Laplacian 
$L_g = - \sum_{i=1}^d \pa_i^2$ the well-known Green's function
\begin{align}
\begin{split}
G_d(x,y) = \begin{cases}
-(2\pi)^{-1}\log|x-y|, & d=2,\\
\frac{1}{(d-2)\om_d \norm{x-y}^{d-2}}, & d \geq 3,
\end{cases}
\end{split}
\label{intro_def_Green}
\end{align}
where $\om_d$ denotes the volume of $S^{d-1}$.
For $d \geq 3$, $G_d(x,y)$ is characterized as the unique fundamental solution vanishing at infinity (see Remark \ref{rem_Green_characterize}).
By Liouville's theorem, this boundary condition is invariant under arbitrary (local) conformal transformations.
In particular, the isomorphism \eqref{eq_intro_psi} becomes natural with respect to (local) conformal transformations (Proposition \ref{prop_conformal_commute}).

In contrast, for $d=2$, the Green's function $G_2(x,y)=-(2\pi)^{-1}\log|x-y|$ admits no such characterization, and the isomorphism \eqref{eq_intro_psi} fails to be natural under (local) conformal transformations.\footnote{
The inverse Fourier transform of the Laplacian $- \sum_{i=1}^d \pa_i^2$ is formally  
$F(\Delta)^{-1} = \frac{1}{(k_1^2+\dots+k_d^2)}$.
In polar coordinates near the origin, one has
$\int \frac{1}{(k_1^2+\dots+k_d^2)}d^dk \sim \int \frac{r^{d-1}}{r^{2}} dr d\om$.
Hence for $d \geq 3$ this expression is locally integrable and defines a distribution, whose inverse Fourier transform yields \eqref{eq_def_Green}.
For $d=2$, however, it is not locally integrable and requires regularization, which cannot be chosen conformally invariantly.
}
In two dimensions, orientation-preserving conformal open embeddings coincide with injective holomorphic maps, and thus, the class of local conformal transformations is extremely large.
In this case, we show that the change in boundary conditions induced by an injective holomorphic map $\phi:\bD_2 \rightarrow \bD_2$ is governed by the harmonic function on $\bD_2 \times \bD_2$ defined by
\begin{align}
H_\phi(z,w) = \log |\phi(z)-\phi(w)| - \log|z-w| -\ft \log|\phi'(z)| - \ft \log|\phi'(w)|,
\label{intro_harm}
\end{align}
(see Theorem~\ref{thm_classical_disk2}).
This harmonic function satisfies the cocycle condition
\begin{align*}
H_{\phi \circ \psi}(z,w)
=H_\psi(z,w)+H_\phi(\psi(z),\psi(w)),
\end{align*}
and defines, by pairing, a second-order differential operator on $\Sym \Cc(\bD)$,
\begin{align*}
\pa_\phi: \Cc(\bD) \otimes \Cc(\bD) \rightarrow \R,\quad 
(f,g)\mapsto \int_{\bD\times \bD} H_\phi(z,w)f(z)g(w) d^2zd^2w.
\end{align*}
This second-order differential operator corresponds to the fact that, in the Fock representation of the affine Heisenberg vertex algebra, the action of the Virasoro algebra (via the Sugawara construction) is given by a quadratic differential operator (see \cite[Remark 3.2]{MBergman}).
Moreover, the cocycle \eqref{intro_harm} vanishes for any M\"{o}bius transformation $\phi(z) = \frac{az+b}{cz+d}$.
Accordingly, although the correlation functions of the prefactorization algebra are not invariant under general local conformal transformations, they become invariant under M\"{o}bius transformations when restricted to the normalized subspace
\begin{align*}
\Cc(\bD_2)_0 = \left\{f \in \Cc(\bD_2) \mid \int_{\bD_2} f d^2z =0\right\}.
\end{align*}
Such a restriction corresponds, in the probabilistic construction of the two-dimensional path integral (the Gaussian Free Field), to working with Sobolev spaces equipped with the Dirichlet inner product \cite{Shef}, and reflects the fact that the two-dimensional massless scalar field theory gives rise to a non-unitary logarithmic CFT \cite{FMS} (see Remark \ref{rem_log}).

%
The central charge is widely understood in terms of the Weyl anomaly, determinant line bundles, the cobordism category, and vertex operator algebras (see \cite{Polyakov,FS,Quillen,BF,Segal,Huang} and the references therein).
Here we explain how it arises from the change of boundary conditions for the Green's function in the framework of factorization algebras.
Using the functor \eqref{intro_HY}, we explicitly describe the action of the flow $\Psi_t^X$ associated with a holomorphic vector field $X=f(z)\pa_z$ 
on the subspace generated by holomorphic point-supported observables, and prove the conformal Ward identity (Theorem \ref{thm_Ward_identity}). Namely, we show that the infinitesimal action of the flow $\Psi_t^X$ is given by the residue of $f(z)$ times the factorization product with a field $T(z)$, thereby identifying $T(z)$ as the {\it holomorphic energy-momentum tensor}.
For an injective holomorphic map
$\phi:\bD\to\bD$, the anomalous term in the conformal transformation
of $T(z)$ is obtained from the harmonic cocycle  as
\begin{align*}
\pa_{z}\pa_{w}
H_{\phi}(z,w)|_{z=w}
=
\frac{1}{12}S(\phi)(z),
\end{align*}
where $S(\phi)(z)$ is the Schwarzian derivative. Comparing this with the standard transformation law of the holomorphic energy-momentum tensor gives central charge $c=1$. Moreover, differentiating this relation along the flow $\Psi_t^X$ gives
\begin{align}
\left.\frac{d}{dt}\right|_{t=0}
\left.
\pa_{z}\pa_{w}
H_{\Psi_{t}^X}(z,w)
\right|_{z=w}
=
\frac{1}{12}f'''(z).\label{intro_GF_add}
\end{align}
The map $f(z)\pa_z\mapsto \frac1{12}f'''(z)\,dz^2$ is the holomorphic restriction of the Gelfand--Fuchs $1$-cocycle, and its meromorphic extension yields the $2$-cocycle
defining the Virasoro algebra (see Corollary~\ref{cor_TT_OPE} and
Remark~\ref{rem_GF}).

\vspace{3mm}

\noindent
\begin{center}
{0.3. \bf 
Observable and harmonic distributions
}
\end{center}
Via the linear isomorphism
\begin{align*}
\Psi_G:\HY(U,g_\std) \overset{\cong}{\rightarrow} 
\mathrm{Sym}(\mathrm{cok}(\Delta:\Cc(U) \rightarrow \Cc(U)))
\label{eq_intro_psi2}
\end{align*}
determined by the Green's function \eqref{intro_def_Green}, we are led to study the cokernel of the Laplacian, $\Cc(U) / \Delta \Cc(U)$.
Using harmonic analysis, we give an explicit description of the induced $\CE$-algebra structure on this space.

Let $\mathrm{H}(U)$ denote the space of harmonic functions on $U$, regarded as a topological vector space with the topology of local uniform $C^\infty$ convergence, and let $\mathrm{H}'(U)$ be its topological dual.
Then the natural pairing
\begin{align*}
\Cc(U) \times \mathrm{H}(U) \rightarrow \R,\quad 
(f,h) \mapsto \int_{U} f(x)h(x) d^dx
\end{align*}
 induces an isomorphism (Theorem \ref{thm_compact_isomorphism})
\begin{align*}
\Cc(U)/ \Delta \Cc(U) \overset{\cong}{\rightarrow} \mathrm{H}'(U).
\end{align*}

Harmonic functions on the disk $B_R(0)$ admit expansions in terms of harmonic polynomials.
We extend this description to harmonic distributions $\mathrm{H}'(B_R(0))$ and give a characterization of the natural embedding (Proposition \ref{prop_equiv_dist})
\[
\mathrm{H}'(B_R(0)) \subset \prod_{n=0}^\infty \Harm_{d,n} \subset \R[[x_1,\dots,x_d]],
\]
where $\Harm_{d,n}$ denotes the space of homogeneous harmonic polynomials of degree $n$ in $d$ variables.
This is analogous to the harmonic polynomial expansion of distributions on the sphere $S^{d-1}$ due to Estrada--Kanwal \cite{EK}.
Furthermore, when $R=1$ and $d \geq 3$, we show that $\Sym H'(\bD)$ embeds naturally and densely into the Hilbert Fock space $\widehat{\Sym} H_\CFT$ (Theorem \ref{thm_subspace}).
Here $H_\CFT$ is a Hilbert space completion of the space of harmonic polynomials $\Harm_d$, and it carries a unitary representation of $\SO^+(d,1)$.
In the case $d=2$, a codimension-one subspace of $H'(\bD_2)$ densely embeds into the direct sum of holomorphic and anti-holomorphic Bergman spaces (see \cite{MLeft,MBergman}).

The space $\mathrm{H}'(\bD)$ contains the evaluation functional at each point $a \in \bD$,
\begin{align*}
\delta_a : \mathrm{H}(\bD) \rightarrow \R,\quad 
u(x) \mapsto u(a),
\end{align*}
and the set $\{\delta_a\}_{a\in \bD} \subset \Hf_\CFT$ spans a dense subspace.
Theorems~\ref{thm_disk_algebra} and~\ref{thm_disk_algebra2} give an explicit description of the $\CE$-algebra structure on this dense subspace.

Whether the algebra structure defined on this dense subspace extends to $\Hf_\CFT$ is a central theme in \cite{MLeft} ($d \geq 3$) and \cite{MBergman} ($d=2$).
When one works with the BV complex built from compactly supported functions, the resulting space of observables does not exhaust the entire $\Hf_\CFT$.
However, in the refined category $\Embc$ introduced in \cite{MLeft}, one may choose a larger space of test functions, such as functions vanishing at infinity.
It remains an important future problem to study whether, by selecting appropriate function spaces and BV complexes, one can construct homologically the Hilbert spaces of axiomatic QFT.


The organization of this paper is as follows.
In Section~\ref{sec_geom}, we review from \cite{MLeft} the definitions and basic properties of the category $\Emb$ and the operad $\CE$.
In Section~\ref{sec_factorization}, we construct the symmetric monoidal functor $\HY$ using the framework of prefactorization algebras.
Section~\ref{sec_classical_quantum} recalls the Costello--Gwilliam construction expressing $\HY$ in terms of the cokernel of the Laplacian via the Green's function on $\R^d$, and studies how this description transforms under conformal open embeddings.
In Section~\ref{sec_central}, we show that in the case $d=2$ a correction term given by a harmonic cocycle is required.
Sections~\ref{sec_compact} and~\ref{sec_expansion} analyze the cokernel of the Laplacian using harmonic analysis and 
characterize their image in $\prod_{n \geq 0}\Harm_{d,n}$.
In Section~\ref{sec_algebra}, we give an explicit description of the $\CE$-algebra structure on the space of distributions.
Finally, in Section~\ref{sec_Ward}, we prove the conformal Ward identity and relate the infinitesimal variation of the harmonic cocycle to the Gelfand--Fuchs cocycle. Appendix~A summarizes basic facts about distributions that are used in Sections~\ref{sec_compact} and~\ref{sec_expansion}.

\section{Geometric settings}\label{sec_geom}
In this section, following \cite[Section 1]{MLeft}, we recall the definition and some basic properties of the category of Riemannian manifolds with conformal open embeddings and of the full monoidal subcategory generated by the unit disk.
Throughout this paper, all manifolds are assumed to be oriented and second countable, and their dimension is fixed to be $d \geq 2$.
Unless otherwise specified, all morphisms between manifolds are assumed to be orientation-preserving
and all vector spaces are taken to be real vector spaces.

\subsection{Category of Riemannian manifolds with conformal open embeddings}
\begin{dfn}\label{def_conformal_map}
A local diffeomorphism  $f : (M_1,g_1) \rightarrow (M_2,g_2)$ between Riemannian manifolds is called a {\bf conformal map} if there exists a smooth map $\Om_f : M_1 \rightarrow \R$ with $\Om_f>0$ such that $f^*(g_2)= \Om_f^2 g_1$.
\end{dfn}

Let $\Emb$ be the category whose objects are oriented Riemannian manifolds $(M,g)$ of dimension $d$ without boundary (not necessarily compact) and whose morphisms are smooth open embeddings that are orientation-preserving conformal maps. We denote the morphisms of this category by $\Emb(M,N)$.
The category $\Emb$ is equipped with a symmetric monoidal structure given by disjoint union of manifolds with the empty set as the unit.
Let $\Vect$ be the symmetric monoidal category of $\R$-vector spaces.
The purpose of this paper is to study a symmetric monoidal functor
\[
\HY: \Emb \rightarrow \Vect,
\]
which is constructed by a prefactorization algebra associated to the conformal Laplacian.

Let us consider $\R^d$ equipped with the standard metric $g_\std = dx_1^2+\dots+dx_d^2$.
Set
\begin{align*}
\bD_d = \{(x_1,\dots,x_d) \in\R^d\mid ||x||<1 \},
\end{align*}
the unit disk equipped with the standard Riemannian metric $g_{\std}$.
We often write $\bD_d$ simply as $\bD$ for brevity.
Let $\tDisk$ denote the full subcategory of $\Emb$ whose objects are disjoint unions of disks $\{\sqcup_n \bD\}_{n \geq 0}$.
The category $\tDisk$ is the full subcategory generated as a monoidal category by the flat disk $\bD_d$.
The following operad is introduced in \cite{MLeft} as an analogue of the little disk operad in conformally flat geometry:
\begin{dfn}\label{def_operad}
For $n \geq 0$, set
\begin{align*}
\CE(n)=\Emb(\underbrace{\bD \sqcup \cdots \sqcup \bD}_n,\bD).
\end{align*}
The collection $\{\CE(n)\}_{n \geq 0}$ has the structure of a permutation operad via the operation of composing into the $i$-th component ($i \in \{1,\dots,n\}$)
\begin{align*}
\circ_i:\CE(n) \times \CE(m) \rightarrow \CE(n+m-1)
\end{align*}
The element $(\id_\bD:\bD \rightarrow \bD)\in \CE(1)$ is the unit, and $\CE(0)$ consists of the unique map $\emptyset \rightarrow \bD$, denoted by $*$. Then, the composition
\begin{align*}
\circ_i:\CE(n) \times \CE(0)=\{*\} \rightarrow \CE(n-1),\quad\phi_{[n]} \mapsto \phi_{[n]}\circ_i *
\end{align*}
is the operation that forgets the $i$-th disk of $\phi_{[n]}$ and regards it as an element of $\CE(n-1)$.
The permutation group $S_n$ acts on $\CE(n)$ by permuting the disks.
\end{dfn}
Note that $\CE(1) = \Emb(\bD,\bD)$ consists of all conformal open embeddings of disks, which need not be surjective.
For example, let $a\in \bD$ and $1>r>0$ satisfy $B_r(a) \subset \bD$,
where $B_r(a)$ denotes the open ball of radius $r$ centered at $a$, $B_r(a)=\{x \in \R^d \mid |x-a|<r\}$.
Then
\begin{align*}
B_{a,r}:\bD \rightarrow \bD,\qquad B_{a,r}(x)= rx+a
\end{align*}
is a conformal open embedding, and hence $B_{a,r} \in \CE(1)$.
More generally, for $a_i \in \bD$ and $1>r_i>0$ ($i=1,\dots,n$) satisfying
\begin{align}
B_{r_i}(a_i) \subset \bD 
\quad\text{and}\quad
B_{r_i}(a_i) \cap B_{r_j}(a_j) =\emptyset 
\qquad (i\neq j),
\label{eq_sep_ball}
\end{align}
we have
\begin{align}
(B_{a_1,r_1},\dots, B_{a_n,r_n}) \in \CE(n).
\label{eq_ball}
\end{align}
%
\begin{dfn}\label{def_operad}
Let $\cC$ be a strict symmetric monoidal category. A $\CE$-algebra in $\cC$ is an operad homomorphism from $\CE$ to the endomorphism operad $\mathrm{End}_A=\{\mathrm{End}_A(n)\}_{n \geq 0}$.
More explicitly, a $\CE$-algebra in $\cC$ is an object $A \in \cC$ equipped with a sequence of maps
\begin{align*}
\rho_n:\CE(n) \rightarrow \Hom_\cC(A^{\otimes n},A)
\end{align*} 
such that:
\begin{enumerate}
\item
$\rho_1(\id_{\bD}) = \id_A$.
\item
For any $\phi_{[n]} \in \CE(n)$ and $\phi_{[m]} \in \CE(m)$ with $n \geq 1, m \geq 0$ and
$i \in \{1,\dots,n\}$,
\begin{align*}
\rho_{n+m-1}(\phi_{[n]}\circ_i \phi_{[m]})=
\rho_{n}(\phi_{[n]})\circ_i \rho_m(\phi_{[m]}).
\end{align*}
\item
The map $\rho_n:\CE(n) \rightarrow \Hom(A^{\otimes n},A)$ is $S_n$-equivariant for all $n \geq 1$.
\end{enumerate}
\end{dfn}
By \cite[Proposition 1.11]{MLeft}, we have:
\begin{prop}\label{prop_CF_algebra}
There is a one-to-one correspondence between $\CE$-algebras in $\cC$ and symmetric monoidal functors $F: \tDisk \rightarrow \cC$.
\end{prop}
\begin{rem}\label{rem_monoidal_envelop}
Proposition \ref{prop_CF_algebra} says precisely that $\tDisk$ is the symmetric monoidal envelope of the operad $\CE$.
\end{rem}

Hereafter in this section, we recall the structure of the operad $\CE$ (see \cite[Section 1.2 and 1.3]{MLeft} for more details).
Let $S^d=\R^d \cup \{\infty\}$ be the $d$-dimensional sphere equipped with the standard Riemannian metric.
The group of orientation-preserving conformal diffeomorphisms of $S^d$, $\Conf^+(S^d)$, is isomorphic to $\SO^+(d+1,1)$. The action of $\Conf^+(S^d)$ on $\R^d\subset S^d$ via stereographic projection is generated by the following elements:
\begin{itemize}
\item
a translation $T_a:x \mapsto x+a$, $a \in \R^d$,
\item
an orthogonal transformation $x \mapsto \Lambda x$, where $\Lambda$ is in the special orthogonal group $\mathrm{SO}(d)$,
\item
a dilation $R^{D}: x \mapsto R x$, $R \in \R_{>0}$, and
\item
a special conformal transformation
\begin{align*}
K_b:x \mapsto \frac{x-(x,x)b}{1-2(x,b)+(x,x)(b,b)},\quad\quad b\in \R^d.
\end{align*}
\end{itemize}
Note that when $d=2$, if we set $w=b_1+i b_2$ and $z=x_1+ix_2$, then $K_b(z) = \frac{z}{1-\bar{w} z}$, and $\Conf^+(S^2) \cong
\mathrm{SO}^+(3,1) \cong  \mathrm{PSL}_2\C$ with the action on $\C P^1$ given by M\"{o}bius transformations.
The following theorem is a classical result due to Liouville \cite{Liouville}:
\begin{thm}\label{thm_Liouville}
Let $U,V \subset \R^d$ be connected open subsets.
\begin{itemize}
\item
Assume $d \geq 3$. A smooth map $\phi:U \rightarrow V$ is an orientation-preserving conformal map if and only if $\phi$ is a restriction of an element in $\Conf^+(S^d)$.
\item
Assume $d=2$. A smooth map $\phi:U \rightarrow V$ is an orientation-preserving conformal map if and only if $\phi$ is a holomorphic map satisfying $\phi'(z)\neq 0$ for any $z\in U$.
\end{itemize}
\end{thm}

It follows immediately that
\begin{align*}
\Disk(1) =
\begin{cases}
\{\phi \in \Conf^+(S^d) \mid \phi(\bD_d) \subset \bD_d \} & \text{ for }d \geq 3\\
\{\phi:\bD_2 \rightarrow \bD_2 \mid \text{$\phi$ is an injective holomorphic map}\} & \text{ for }d =2.
\end{cases} 
\end{align*}
Hence, we have \cite[Proposition 1.20 and Proposition 1.21]{MLeft}:
\begin{prop}
For each $n>0$, $\Disk(n)$ is in bijection with
\begin{align*}
\{(g_1,\dots,g_n) \in \Disk(1)^n \mid {g_i(\bD)} \cap  {g_j(\bD)} =\emptyset \text{ for any }i\neq j\},
\end{align*}
the set of disjoint conformal open embeddings.
Moreover, the operadic composition $\circ_i:\Disk(n) \times \Disk(m) \rightarrow \Disk(n+m-1)$ is given by
\begin{align*}
(g_1,\dots,g_n) \circ_i (f_1,\dots,f_m) = (g_1,\dots,g_{i-1}, g_i \circ f_1, \dots, g_i\circ f_m, g_{i+1},\dots,g_n).
\end{align*}
Furthermore, if $d \geq 3$, then as monoids,
\begin{align*}
\Disk(1) \cong \{\phi \in \Conf^+(S^d) \mid \phi(\bD) \subset \bD \}.
\end{align*}
\end{prop}

\section{Construction of the symmetric monoidal functor}\label{sec_construction}
In this section, using the Costello-Gwilliam factorization algebra \cite{CG}, we construct a symmetric monoidal functor 
$\HY:\Emb \rightarrow \Vect$ and study its properties.

In Section \ref{sec_factorization}, we construct the functor $\HY$.
In Section \ref{sec_classical_quantum}, following \cite{CG}, we recall an alternative description (the classical description) of the vector space $\HY(U,g_\std)$ for $U \subset \R^d$, using the Green's function on $(\R^d,g_\std)$.
By examining the behavior of the Green's function under conformal transformations, we show that for $d \geq 3$ the classical description defines a natural transformation between the functors.
In Section~\ref{sec_central}, we show that for $d=2$ naturality fails and that a cocycle consisting of harmonic functions appears.


\subsection{Prefactorization algebra of conformal Laplacian}\label{sec_factorization}
Let $(M,g)$ be an oriented $d$-dimensional Riemannian manifold and
\begin{align*}
L_g = \Delta_g +  \frac{d-2}{4(d-1)}S_g,
\end{align*}
the \textbf{conformal Laplacian} (also known as the Yamabe operator \cite{Yamabe}), where $\Delta_g=d^*d$ is the Hodge Laplacian and $S_g$ is the scalar curvature \cite{LP}.
When we consider $\R^d$ equipped with the standard metric $g_\std = dx_1^2+\dots+dx_d^2$,
\begin{align*}
L_{g_\std}= \Delta_{g_\std}= - \sum_{i=1}^d {\pa}_i^2.
\end{align*}

An important property of the conformal Laplacian is the following \cite[(2.7)]{LP}:
\begin{prop}\label{prop_Yamabe}
Let $\om:M \rightarrow \R$ be a smooth map.
Then, for $\hat{g} = e^{2\om} g$,
\begin{align*}
L_{\hat{g}}(f) = e^{-\frac{d+2}{2}\om} L_g (e^{\frac{d-2}{2}\om} f).
\end{align*}
\end{prop}

Set
\begin{align*}
P(C_c^\infty(M))&= \bigoplus_{n \geq 0} \mathrm{Sym}^n C_c^\infty(M),\\
P^n(C_c^\infty(M))&= \mathrm{Sym}^n C_c^\infty(M).
\end{align*}
Here $C_c^\infty(M)$ denotes the vector space of real valued smooth functions on $M$ with compact support.
Throughout this paper, unless otherwise stated, all direct sums, tensor products, $\Sym$, cokernels, etc., are taken in the category of $\R$-vector spaces.
The vector space $P(C_c^\infty(M))$ is a commutative associative algebra, the polynomial algebra on $\Cc(M)$, and $P(C_c^\infty(\emptyset)) = \R$.
For $f_1,\dots,f_n \in \Cc(M)$, denote their product by $f_1\cdot f_2 \cdots f_n$.

Following \cite{CG}, define a linear map $\Delta_{M,g}^\mathrm{BV}:P(C_c^\infty(M))\otimes C_c^\infty(M) \rightarrow P(C_c^\infty(M))$ by
\begin{align}
\Delta_{M,g}^{\mathrm{BV}}(f_1 \cdots f_n \otimes h)=- f_1 \cdots f_n \cdot L_g h+\sum_{i=1}^n \left(\int_{M} f_i h \mathrm{vol}_{g} \right) f_1 \cdots  f_{i-1} \cdot f_{i+1} \cdots f_n,
\label{eq_def_BV}
\end{align}
where $\mathrm{vol}_{g}$ is the volume form of $(M,g)$, and set
\begin{align*}
\HY(M,g) = \mathrm{cok}(\Delta_{M,g}^\mathrm{BV}:  P(C_c^\infty(M)) \otimes C_c^\infty(M) \rightarrow P(C_c^\infty(M))) = P(C_c^\infty(M)))/ \mathrm{Im}(\Delta_{M,g}^\mathrm{BV}).
\end{align*}
We refer to Remark \ref{rem_formal_parameter} for an explanation of why the Planck constant $\hbar$ does not appear in the definition above.



Define a functor
$\HY:\Emb \rightarrow \Vect$ as follows:
On objects, we send a Riemannian manifold $(M,g)$ to the vector space $\HY(M,g)$.
Let $\phi : (M_1,g_1) \rightarrow (M_2,g_2)$ be a conformal open embedding with the conformal factor $\Om_\phi^2 g_1 = \phi^*(g_2)$.
Motivated by the conformal covariance in Proposition \ref{prop_Yamabe}, for $f \in \Cc(M_1)$, define a function $W_\phi (f):M_2 \rightarrow \R$ by
\begin{align}
\begin{split}
W_\phi (f)(x) = \begin{cases}
\Om_\phi(\phi^{-1}(x))^{-\frac{d+2}{2}} f(\phi^{-1}(x)), & x \in \phi(M_1) \\
0, & \text{otherwise}
\end{cases}
\end{split}
\label{eq_def_W}
\end{align}
for $x\in M_2$. 
Since $f$ has compact support in $M_1$ and $\phi$ is a smooth open embedding, 
$W_\phi(f) \in \Cc(M_2)$.
Similarly, we define $\tW_\phi(f) \in \Cc(M_2)$ by
\begin{align*}
\tW_\phi(f)(x) 
= \Om_\phi(\phi^{-1}(x))^{-\frac{d-2}{2}} 
  f(\phi^{-1}(x)).
\end{align*}


Define a linear map
\begin{align*}
W_\phi^0: P(\Cc(M_1)) \rightarrow P(\Cc(M_2))
\end{align*}
by $W_\phi^0(f_1 \cdots f_n) = W_\phi(f_1)\cdots W_\phi(f_n)$ and
\begin{align*}
W_\phi^1:P(\Cc(M_1))\otimes \Cc(M_1) \rightarrow P(\Cc(M_2))\otimes \Cc(M_2)
\end{align*}
by $W_\phi^1(f_1 \cdots f_n \otimes h) = W_\phi(f_1)\cdots W_\phi(f_n) \otimes \tW_\phi(h)$.
\begin{lem}\label{lem_func_well}
Let $\phi : (M_1,g_1) \rightarrow (M_2,g_2)$ be a conformal open embedding. Then, the following diagram commutes:
\begin{align*}
\begin{array}{ccc}
P(\Cc(M_1))\otimes \Cc(M_1) &\overset{\Delta_{M_1,g_1}^{\mathrm{BV}}}{\longrightarrow}&P(\Cc(M_1)) \\
\downarrow_{W_\phi^1} & & \downarrow_{W_\phi^0} \\
P(\Cc(M_2))\otimes \Cc(M_2) &\overset{\Delta_{M_2,g_2}^{\mathrm{BV}}}{\longrightarrow}&P(\Cc(M_2)).
\end{array}
\end{align*}
\end{lem}
\begin{proof}
By definition, we have
\begin{align*}
&\Delta_{M_2,g_2}^\mathrm{BV} \circ W_\phi^1(f_1 \cdots  f_n \otimes h)=\Delta_{M_2,g_2}^\mathrm{BV} (W_\phi(f_1) \cdots  W_\phi(f_n) \otimes \tW_\phi(h))\\
&=  \sum_{i=1}^n \left(\int_{M_2} W_\phi(f_i) \tW_\phi(h)\mathrm{vol}_{g_2} \right) (W_\phi(f_1) \cdots  \hat{f_i}  \cdots  W_\phi(f_n)) - (W_\phi(f_1) \cdots W_\phi(f_n) \cdot (L_{M_2, g_2} \tW_\phi(h)).
\end{align*}
By setting $y=\phi(x)$, we have:
\begin{align*}
\int_{M_2} W_\phi(f_i) \tW_\phi(h)\mathrm{vol}_{g_2}&= \int_{M_2} \Om_\phi(\phi^{-1}(y))^{-\frac{d+2}{2}} f(\phi^{-1}(y))
 \Om_\phi(\phi^{-1}(y))^{-\frac{d-2}{2}}h(\phi^{-1}(y)) \mathrm{vol}_{g_2}(y)\\
 &= \int_{\phi(M_1)} \Om_\phi(\phi^{-1}(y))^{-d} f(\phi^{-1}(y))h(\phi^{-1}(y)) \mathrm{vol}_{g_2}(y)\\
 &= \int_{M_1} \Om_\phi(x)^{-d} f(x)h(x) \mathrm{vol}_{\phi^*(g_2)}(x)\\
  &= \int_{M_1} \Om_\phi(x)^{-d} f(x)h(x) \mathrm{vol}_{\Om_\phi^2 g_1}(x)\\
   &= \int_{M_1} f(x)h(x) \mathrm{vol}_{g_1}(x).
\end{align*}
By Proposition \ref{prop_Yamabe}, we have:
\begin{align*}
L_{M_2, g_2} \tW_\phi(h)&= L_{M_2, g_2}  \Om_\phi(\phi^{-1}(y))^{-\frac{d-2}{2}}h(\phi^{-1}(y))=L_{\phi(M_1), g_2}  \Om_\phi(\phi^{-1}(y))^{-\frac{d-2}{2}}h(\phi^{-1}(y))\\
&=\left(L_{M_1, \phi^*(g_2)}  \Om_\phi^{-\frac{d-2}{2}} h\right)\circ \phi^{-1}=\left(L_{M_1, \Om_\phi g_1}  \Om_\phi^{-\frac{d-2}{2}} h\right)\circ \phi^{-1}\\
&=\left(  \Om_\phi^{-\frac{d+2}{2}} L_{M_1, g_1}h\right)\circ \phi^{-1}=W_\phi (L_{M_1,g_1}(h)).
\end{align*}
Hence, the assertion holds.
%
\end{proof}

Let $\psi:(M_2,g_2) \rightarrow (M_3,g_3)$ be a conformal open embedding.
Then, for $f\in \Cc(M_1)$,
\begin{align*}
W_\psi W_\phi(f) &= W_\psi\left( (\Om_{\phi}\circ \phi^{-1})^{-\frac{d+2}{2}}f \circ \phi^{-1}\right)\\
&= (\Om_{\psi} \circ \psi^{-1})^{-\frac{d+2}{2}} (\Om_{\phi}\circ \phi^{-1}\circ \psi^{-1})^{-\frac{d+2}{2}}f \circ \phi^{-1} \circ \psi^{-1}\\
&= (\Om_{\psi\circ \phi} \circ  \phi^{-1} \circ \psi^{-1})^{-\frac{d+2}{2}}f \circ \phi^{-1} \circ \psi^{-1}=W_{\psi\circ \phi}(f)
\end{align*}
by $\phi^*\psi^*(g_3)=\phi^*(\Om_{\psi}^2g_2) = \Om_\phi^2 \cdot(\Om_\psi \circ \phi)^2 g_1$.
Hence, $\HY:\Emb \rightarrow \Vect$ is a functor.

We claim that $\HY$ is symmetric monoidal.
In fact, for $(M_1,g_1), (M_2,g_2) \in \Emb$,
\begin{align*}
\Cc(M_1 \sqcup M_2) \cong \Cc(M_1) \oplus \Cc(M_2)
\end{align*}
and thus $P(\Cc(M_1 \sqcup M_2)) \cong P(\Cc(M_1)) \otimes P(\Cc(M_2))$, which is functorial.
Note that for any $f^{(1)} \in \Cc(M_1)$ and $f^{(2)} \in \Cc(M_2)$
\begin{align*}
\int_{M_1 \sqcup M_2} f^{(1)} f^{(2)} \vol_{g_1 \sqcup g_2}=0.
\end{align*}
Hence, for any $(f_1^{(1)} \cdots f_n^{(1)}) \otimes (f_1^{(2)} \cdots f_m^{(2)})\otimes (h_1+h_2) \in 
P(\Cc(M_1 \sqcup M_2)) \otimes \Cc(M_1 \sqcup M_2) \cong (P(\Cc(M_1)) \otimes P(\Cc(M_2)))\otimes (\Cc(M_1)\oplus \Cc(M_2))$,
we have
\begin{align*}
&\Delta_{M_1 \sqcup M_2,g_1\sqcup g_2}^\mathrm{BV}((f_1^{(1)} \cdots f_n^{(1)}) \otimes (f_1^{(2)} \cdots f_m^{(2)}) \otimes (h_1+h_2))\\
&=
\Delta_{M_1,g_1}^\mathrm{BV}((f_1^{(1)} \cdots f_n^{(1)})\otimes h_1)
\otimes (f_1^{(2)} \cdots f_m^{(2)})
+
(f_1^{(1)} \cdots f_n^{(1)})\otimes \Delta_{M_2,g_2}^\mathrm{BV}((f_1^{(2)} \cdots f_m^{(2)})\otimes h_2).
\end{align*}
Hence, we have:
\begin{thm}\label{prop_functor}
For any $d\geq 2$, $\HY:\Emb \rightarrow \Vect$ is a symmetric monoidal functor.
\end{thm}

By Proposition \ref{prop_CF_algebra}, we have:
\begin{cor}\label{cor_disk_algebra}
The restriction of $\HY$ to $\tDisk$ determines a $\CE$-algebra structure on the vector space $\HY(\bD)$.
\end{cor}

\begin{rem}\label{rem_formal_parameter}
We have set the Planck constant $\hbar$ equal to $1$.
For any fixed $\hbar>0$, define
\begin{align*}
\Delta_{M,g,\hbar}^{\mathrm{BV}}(f_1\cdots f_n\otimes h)
=
-f_1\cdots f_n\cdot L_g h
+
\hbar\sum_{i=1}^n
\left(\int_M f_i h\,\mathrm{vol}_g\right)
f_1\cdots f_{i-1}f_{i+1}\cdots f_n.
\end{align*}
If $a^2=\hbar$, multiplication by $a^n$ on
$\Sym^n C_c^\infty(M)$ and by $a^{n+1}$ on
$\Sym^n C_c^\infty(M)\otimes C_c^\infty(M)$ intertwines
$\Delta_{M,g,\hbar}^{\mathrm{BV}}$ with
$\Delta_{M,g,1}^{\mathrm{BV}}$
and is compatible with conformal open embeddings.
Hence we may set $\hbar=1$ throughout.
\end{rem}

\subsection{Quantization and conformal transformations}\label{sec_classical_quantum}
In this section we consider the restriction of the functor $\HY:\Emb \rightarrow \Vect$ to open subsets of $(\R^d,g_\std)$.
By \cite[Section 4.6.2]{CG}, we have a linear isomorphism
\begin{align*}
\Psi_{G_d}(U): \HY(U) \cong P(\Cc(U)/\Delta \Cc(U))
\end{align*}
by using the Green function
\begin{align}
\begin{split}
G_d(x,y) = \begin{cases}
-(2\pi)^{-1}\log|x-y|, & d=2\\
\frac{1}{(d-2)\om_d \norm{x-y}^{d-2}} & d \geq 3
\end{cases}
\end{split}
\label{eq_def_Green}
\end{align}
for each open subset $U \subset \R^d$. Here, $\om_d$ is the volume of the $(d-1)$-dimensional sphere $S^{d-1} \subset \R^d$ induced from $g_\std$.
The purpose of this section is to show that $\Psi_{G_d}(U)$ is a natural transformation with respect to arbitrary conformal maps when $d \geq 3$, but it is not natural in the case $d=2$.
This is due to the fact that the Green function in two dimensions does not vanish at infinity. 
Since $\Psi_{G_d}(U)$ yields a physical quantity called the {\bf vacuum expectation value}, the failure of the naturality corresponds to the fact that the local conformal invariance of the vacuum expectation value is broken in $d=2$. This conformal anomaly is well known in physics \cite{Polyakov,Duff}.
The following result is basic:
\begin{prop}\cite[Proposition 1.18]{MLeft}
\label{prop_conf_identity}
For any $d \geq 2$ and $\phi \in \Conf^+(S^d)$, $\norm{\phi(x)-\phi(y)}^2=\Om_\phi(x)\Om_\phi(y)\norm{x-y}^2$ holds.
In particular, if $d \geq 3$, then for $\phi \in \Conf^+(S^d)$
\begin{align}
G_d(x,y)= \Om_\phi(x)^{\frac{d-2}{2}}\Om_\phi(y)^{\frac{d-2}{2}}
G_d(\phi(x),\phi(y)).\label{eq_Green_trans}
\end{align}
\end{prop}
Note that this proposition no longer holds for a general $\phi\in \CE(1)$ when $d=2$.

\begin{rem}[Characterization of the Green function]
\label{rem_Green_characterize}
Using the fact that any bounded harmonic function must be constant, the fundamental solution $\Psi_d(x) = \frac{1}{(d-2)\om_d \norm{x}^{d-2}}$ for $d \geq 3$ is characterized as the unique function satisfying the following conditions:
\begin{enumerate}
\item
$F(x) \in L_{\mathrm{loc}}^1(\R^d)$, a locally-integrable function, and $F|_{\R^d \setminus \{0\}}$ is continuous.
\item
$\Delta_{g_\std} F(x) = \delta_0$ as a distribution.
\item
For any $\epsilon >0$, there is $R>0$ such that $\sup_{|x|>R}|F(x)|<\ep$.
\end{enumerate}
This fact follows easily from \cite[Theorem 3.3.2]{Hormander} and Weyl's lemma (see for example \cite[Theorem IX.25]{RS2}).
On the other hand, in $d=2$, since $\log|z|$ does not vanish at infinity, it does not satisfy (3).
Thus, the properties of the fundamental solution differ between $d \geq 3$ and $d=2$.
\end{rem}

%
%

Let $U \subset \R^d$.
We recall the construction of the linear isomorphism from \cite[Section 4.6.2]{CG}
\begin{align*}
\Psi_{G_d}(U): \HY(U) = P(\Cc(U))/\Delta_{\R^d,g_\std}^{\mathrm{BV}} \rightarrow P(\Cc(U)/\Delta \Cc(U)).
\end{align*}
Let $P \in L_{\mathrm{loc}}^1(U \times U)$, a locally-integrable function, with $P(x,y)=P(y,x)$. Then, define a symmetric bilinear form $\partial_P :\Cc(U) \otimes \Cc(U) \rightarrow \R$
by 
\begin{align*}
\pa_P (f,g) = \int_{U \times U} P(x,y)f(x)g(y)dxdy.
\end{align*}
Then, we can extend $\pa_P$ on $P^n(\Cc(U)) = \mathrm{Sym}^n \Cc(U)$ by
\begin{align*}
\pa_P(f_1\dots f_n) = \sum_{1 \leq i <j \leq n}\pa_P(f_i,f_j) f_1\dots \hat{f_i}\dots\hat{f_j}\dots f_n,
\end{align*}
which is a linear map $P^n(\Cc(U)) \rightarrow P^{n-2}(\Cc(U))$.
Here $\pa_P$ lowers the degree, so $\exp(\pa_P):P(\Cc(U)) \rightarrow P(\Cc(U))$ is a well-defined linear map.
Similarly, on $P(\Cc(U)) \otimes \Cc(U)$, by letting $\pa_P$ act only on the $P(\Cc(U))$ factor, we obtain a linear map $\exp(\pa_P):P(\Cc(U))\otimes \Cc(U) \rightarrow P(\Cc(U))\otimes \Cc(U)$.
Let $\Delta_{\R^d,g_\std}^{\mathrm{cl}}:P(\Cc(U))\otimes \Cc(U) \rightarrow P(\Cc(U))$ be a linear map defined by
\begin{align*}
\Delta_{\R^d,g_\std}^{\mathrm{cl}}(f_1\dots f_n\otimes h) =-f_1\dots f_n\cdot \Delta_{g_\std} h.
\end{align*}
Then, it is clear that $P(\Cc(U))/\mathrm{Im}(\Delta_{\R^d,g_\std}^{\mathrm{cl}})$ is isomorphic to
$P(\Cc(U)/\Delta_{g_\std} \Cc(U))$, the polynomial ring over the vector space $\Cc(U)/\Delta_{g_\std} \Cc(U)$.
Set
\begin{align*}
S^\cl(U) =\Cc(U)/\Delta_{g_\std} \Cc(U)
\end{align*}
$P(S^\cl(U))$ is called the space of \textbf{classical observables},
while $P(\Cc(U))/\Delta_{\R^d,g_\std}^{\mathrm{BV}}$ is called the space of quantum observables.
The following proposition is proved in \cite[Section 4.6.2]{CG}.
\begin{prop}\cite[Section 4.6.2]{CG}\label{prop_CG}
Let $P(x,y) \in L_{\mathrm{loc}}^1(U\times U)$ satisfy $P(x,y)=P(y,x)$ and $\Delta_{g_\std,x} P(x,y) =\delta(x-y)$ as distribution. Then, the following diagram commutes:
\begin{align*}
\begin{array}{ccc}
P(\Cc(U)) \otimes \Cc(U) &\overset{\Delta_{\R^d,g_\std}^{\mathrm{BV}}}{\longrightarrow}& P(\Cc(U))\\
\downarrow_{\exp(\pa_{P})} & & \downarrow_{\exp(\pa_{P})} \\
P(\Cc(U)) \otimes \Cc(U) &\overset{\Delta_{\R^d,g_\std}^{\mathrm{cl}}}{\longrightarrow}& P(\Cc(U)).
\end{array}
\end{align*}
Hence, $\exp(\pa_{P})$ induces a linear isomorphism
\begin{align*}
\Psi_P(U)= \exp(\pa_P): P(\Cc(U))/ \Delta_{\R^d,g_\std}^{\mathrm{BV}} \cong P(S^\cl(U)).
\end{align*}
\end{prop}
\begin{proof}
Since, for any $f_1,f_2 \in \Cc(U)$,
\begin{align*}
\pa_P(f_1,\Delta f_2) &= \int_{U\times U}
P(x,y)f_1(x)(\Delta_y f_2(y)) dxdy= \int_{U\times U}
(\Delta_y P(x,y))f_1(x) f_2(y) dxdy\\
&= \int_{U\times U}
\delta(x-y)f_1(x) f_2(y) dxdy= \int_{U}f_1(x) f_2(x) dx,
\end{align*}
we have $\partial_P^k(f_1\cdots f_n\Delta h)=
(\partial_P^k(f_1\cdots f_n))\Delta h+
k\sum_{i=1}^n
\left(\int_U f_i(x)h(x)\,dx\right)
\partial_P^{k-1}(f_1\cdots \widehat{f_i}\cdots f_n)$
for $k\geq 1$, which implies
\begin{align*}
\exp(\partial_P)(f_1\cdots f_n\Delta h)
=\exp(\partial_P)(f_1\cdots f_n)\Delta h+\sum_{i=1}^n
\left(\int_U f_i(x)h(x)\,dx\right)
\exp(\partial_P)(f_1\cdots \widehat{f_i}\cdots f_n),
\end{align*}
and thus, $\exp(\partial_P)\circ
\Delta^{\mathrm{BV}}_{\mathbb R^d,g_{\rm std}}
=
\Delta^{\mathrm{cl}}_{\mathbb R^d,g_{\rm std}}
\circ\exp(\partial_P)$ holds.
%
\end{proof}

For each open subset $U \subset \R^d$, we write simply $\Psi_G(U)$ for the isomorphism in Proposition~\ref{prop_CG} defined by the restriction of the Green function $G_d(x,y)|_{U \times U}$, \eqref{eq_def_Green}.

Define a $\Disk$-algebra structure on $P(S^\cl(\bD))$
\begin{align*}
\m_{\bullet}^\cl:\Disk(n) \rightarrow \Hom_\Vect(P(S^\cl(\bD))^{\otimes n}, P(S^\cl(\bD))
\end{align*}
as follows:
For $\phi_{[n]}=(\phi_1,\dots,\phi_n) \in \Disk(n)$,
\begin{align}
\m_{\phi_{[n]}}^\cl:P(S^\cl(\bD))^{\otimes n} \overset{\otimes_{i=1}^n \Psi_{G}(\bD)^{-1}}{\longrightarrow} \HY(\bD)^{\otimes n} \overset{\HY(\phi_{[n]})}{\rightarrow} \HY(\bD) 
 \overset{\Psi_G(\bD)}{\rightarrow} P(S^\cl(\bD)).\label{def_pullback_algebra}
\end{align}

In what follows, we study this $\CE$-algebra structure.
Let $\phi:\bD \rightarrow \bD$ be a conformal open embedding.
Recall that the linear map $W_\phi^0: P(\Cc(\bD)) \rightarrow P(\Cc(\bD))$ is given by
\begin{align*}
f_1\cdots f_n &\mapsto ((\Om_\phi(\phi^{-1}(x))^{-\frac{d+2}{2}} f_1\circ \phi^{-1})\cdots ((\Om_\phi(\phi^{-1}(x))^{-\frac{d+2}{2}} f_n\circ \phi^{-1}),
\end{align*}
which induces a linear map $\HY(\phi): \HY(\bD) \rightarrow \HY(\bD)$.

Similarly, for any $f \in \Cc(\bD)$, since by Proposition \ref{prop_Yamabe},
\begin{align}
\Om_\phi(\phi^{-1}(x))^{-\frac{d+2}{2}} &(\Delta_{g_\std}f)\circ \phi^{-1}(x)=\left(\Om_\phi(x)^{-\frac{d+2}{2}} L_{g_\std} f \right)\circ \phi^{-1}=\left(L_{\phi^* g_\std} \Om_\phi(x)^{-\frac{d-2}{2}} f \right)\circ \phi^{-1}\nonumber \\
&=L_{g_\std} \left(\Om_\phi(\phi^{-1}(x))^{-\frac{d-2}{2}} f\circ \phi^{-1}\right)
=\Delta_{g_\std} \left(\Om_\phi(\phi^{-1}(x))^{-\frac{d-2}{2}} f\circ \phi^{-1}\right),
\label{eq_def_W_add2}
\end{align}
$W_\phi^0$ induces a linear map $W_\phi^{\mathrm{cl}}: P(S^\cl(\bD)) \rightarrow P(S^\cl(\bD))$.
\begin{prop}\label{prop_conformal_commute}For $d \geq 3$ and a conformal open embedding $\phi:\bD \rightarrow \bD$, the following diagram commutes:
\begin{align*}
\begin{array}{ccc}
P(\Cc(\bD))/\Delta_{\R^d,g_\std}^{\mathrm{BV}}
&\overset{\HY(\phi)}{\longrightarrow}& P(\Cc(\bD))/\Delta_{\R^d,g_\std}^{\mathrm{BV}}\\
\downarrow_{\Psi_G(\bD)} & & \downarrow_{\Psi_G(\bD)} \\
P(\Cc(\bD)/\Delta_{\R^d,g_\std} \Cc(\bD)) &\overset{W_\phi^{\mathrm{cl}}}{\longrightarrow}& P(\Cc(\bD)/\Delta_{\R^d,g_\std} \Cc(\bD)).
\end{array}
\end{align*}
\end{prop}
\begin{proof}
Let $f_1,f_2 \in \Cc(\bD)$. Since the support of $f_i \circ \phi^{-1}$ is in ${\phi(\bD)} \subset \bD$, by Proposition \ref{prop_conf_identity} and
Theorem \ref{thm_Liouville} and setting $x=\phi(x')$ and $y=\phi(y')$, we have
\begin{align}
\begin{split}
&\pa_{G_d}(W_\phi(f_1), W_\phi(f_2))\\
&=\pa_{G_d}((\Om_\phi(\phi^{-1}(x))^{-\frac{d+2}{2}} f_1\circ \phi^{-1}), (\Om_\phi(\phi^{-1}(x))^{-\frac{d+2}{2}} f_2\circ \phi^{-1}))\\
&=\int_{\bD \times \bD}G_d(x,y)\Om_\phi(\phi^{-1}(x))^{-\frac{d+2}{2}} f_1\circ \phi^{-1}(x) \Om_\phi(\phi^{-1}(y))^{-\frac{d+2}{2}} f_2\circ \phi^{-1}(y) \mathrm{vol}_{g_\std}(x) \mathrm{vol}_{g_\std}(y)\\
&=\int_{\phi(\bD) \times \phi(\bD)}G_d(x,y)\Om_\phi(\phi^{-1}(x))^{-\frac{d+2}{2}} f_1\circ \phi^{-1}(x) \Om_\phi(\phi^{-1}(y))^{-\frac{d+2}{2}} f_2\circ \phi^{-1}(y) \mathrm{vol}_{g_\std}(x) \mathrm{vol}_{g_\std}(y)\\
&=\int_{\bD \times \bD}G_d(\phi(x'),\phi(y'))\Om_\phi(x')^{-\frac{d+2}{2}} f_1(x') \Om_\phi(y')^{-\frac{d+2}{2}} f_2(y') \mathrm{vol}_{\phi^*g_\std}(x') \mathrm{vol}_{\phi^* g_\std}(y')\\
&=\int_{\bD\times \bD}G_d(\phi(x'),\phi(y'))\Om_\phi(x')^{\frac{d-2}{2}} f_1(x') \Om_\phi(y')^{\frac{d-2}{2}} f_2(y') \mathrm{vol}_{g_\std}(x') \mathrm{vol}_{g_\std}(y')\\
&=\int_{\bD \times \bD}G_d(x',y')f_1(x')f_2(y') \mathrm{vol}_{g_\std}(x') \mathrm{vol}_{g_\std}(y')=\pa_{G_d}(f_1,f_2).
\end{split}
\label{eq_G_d_comp}
\end{align}
Thus, for any $f_1,\dots,f_n \in \Cc(\bD)$, 
\begin{align*}
\exp(\pa_{G_d})\circ \HY(\phi)(f_1 \cdots  f_n)
&= \exp(\pa_{G_d})(W_\phi(f_1) \cdots  W_\phi(f_n))\\
&= W_\phi^\cl \circ \exp(\pa_{G_d})(f_1 \cdots  f_n).
\end{align*}
Hence, the assertion holds.
\end{proof}
It is important to note that this naturality fails in $d=2$. 
We will discuss this issue in the next section.

Let $d \geq 3$.
By Proposition \ref{prop_conformal_commute}, for any $\phi \in \Disk(1)$,
\begin{align*}
\m_\phi^{\mathrm{cl}} = W_\phi^{\mathrm{cl}}
\end{align*}
holds. In particular, $W_\bullet^\cl: \Disk(1) \rightarrow \mathrm{End}(P(S^\cl(\bD)))$ is a monoid homomorphism.

Next, we write the higher products explicitly.
Let $(\phi_1,\phi_2) \in \CE(2)$ and 
\[
F^{(1)}=[f_1^{(1)}]\cdots [f_n^{(1)}], \quad  
F^{(2)}=[f_1^{(2)}]\cdots [f_m^{(2)}] 
\in P(S^\cl(\bD)),
\]
where $f_i^{(k)} \in \Cc(\bD)$ is a representative of $[f_i^{(k)}] \in S^\cl(\bD)$.
Then, by definition of $\HY(\phi_1,\phi_2)$ and \eqref{eq_G_d_comp}, we have
\begin{align}
\begin{split}
\m_{\phi_1,\phi_2}^\cl(F^{(1)},F^{(2)})
&=\Psi_G(\bD)\HY(\phi_1,\phi_2)
\left(\Psi_G(\bD)^{-1}(F^{(1)}),
      \Psi_G(\bD)^{-1}(F^{(2)})\right),\\
&=\exp(\pa_{G_d})\left(\left(W_{\phi_1}\circ \exp(-\pa_{G_d})(F^{(1)})\right)\cdot \left(W_{\phi_2}\circ \exp(-\pa_{G_d})(F^{(2)})\right)\right)\\
&=\exp(\pa_{G_d})\left(\left(\exp(-\pa_{G_d})(W_{\phi_1} F^{(1)})\right)\cdot \left(\exp(-\pa_{G_d})(W_{\phi_2} F^{(2)})\right)\right).
\end{split}
\label{eq_why_cont}
\end{align}
The contractions among the factors of
$W_{\phi_1}F^{(1)} = [W_{\phi_1}(f_1^{(1)})]\cdots [W_{\phi_1}(f_n^{(1)})]$
arising from $\exp(-\pa_{G_d})$ are canceled by the outer operator $\exp(\pa_{G_d})$.
The same applies to $\exp(-\pa_{G_d})(W_{\phi_2}F^{(2)})$.
Thus, only the contractions via $\pa_{G_d}$ between factors of
$W_{\phi_1}F^{(1)}$ and $W_{\phi_2}F^{(2)}$ remain.
Hence,
\begin{align*}
\m_{\phi_1,\phi_2}^\cl(F^{(1)},F^{(2)})&=W_{\phi_1}^\cl(F^{(1)})\cdot W_{\phi_2}^\cl(F^{(2)})
 + \sum \text{contractions between }
   W_{\phi_1}^\cl(F^{(1)}) 
   \text{ and } 
   W_{\phi_2}^\cl(F^{(2)}).
\end{align*}
More precisely, by contractions between $W_{\phi_1}^\cl(F^{(1)})$ and $W_{\phi_2}^\cl(F^{(2)})$, we mean all possible pairings of subsets of $\{1,\dots,n\}$ and $\{1,\dots,m\}$ of arbitrary length, contracted via $\pa_{G_d}$.

%


For $N>0$, set $[N]=\{1,\dots,N\}$.
Writing the contractions explicitly, we obtain
\begin{align}
\begin{split}
\m_{\phi_1,\phi_2}^\cl(F^{(1)},F^{(2)})
&=\sum_{p=0}^{\min\{n,m\}}
\frac{1}{p!}
\sum_{\substack{C_1:[p]\hookrightarrow[n] \\ C_2:[p]\hookrightarrow[m]}}
\prod_{k \in [p]}
\pa_{G_d}\bigl(
  W_{\phi_1}(f_{C_1(k)}^{(1)}), 
  W_{\phi_2}(f_{C_2(k)}^{(2)})
\bigr)\\
&\quad\times
\prod_{i \in [n] \setminus C_1([p])}
W_{\phi_1}(f_i^{(1)})
\prod_{j \in [m] \setminus C_2([p])}
W_{\phi_2}(f_j^{(2)}),
\end{split}
\label{eq_Wick_general_formula_add}
\end{align}
where $C_1:[p] \hookrightarrow [n]$ and $C_2:[p] \hookrightarrow [m]$ range over all injective maps.
%
For example,
\begin{align*}
\m_{\phi_1,\phi_2}^\cl(f_1f_2,h)&=W_{\phi_1}^\cl(f_1)W_{\phi_1}^\cl(f_2)W_{\phi_2}^\cl(h)\\
&+\pa_{G_d}\left(W_{\phi_1}^\cl(f_2),W_{\phi_2}^\cl(h)\right) W_{\phi_1}^\cl(f_1)+
\pa_{G_d}\left(W_{\phi_1}^\cl(f_1),W_{\phi_2}^\cl(h)\right) W_{\phi_1}^\cl(f_2).
\end{align*}
%
%
Similarly, for $\phi_{[n]} = (\phi_1,\dots,\phi_n) \in \Disk(n)$, 
the multiplication is defined by contractions in the same way.
\begin{thm}\label{thm_classical_disk}
For any $d \geq 3$, the $\Disk$-algebra
$P(S^\cl(\bD))$
satisfies
\begin{itemize}
\item
For any $\phi \in \Disk(1)$,
 $\m_\phi^{\mathrm{cl}} = W_\phi^{\mathrm{cl}}$. In particular,
 $W_\bullet^\cl: \Disk(1) \rightarrow \mathrm{End}(P(S^\cl(\bD)))$ is a monoid homomorphism.
 \item
 For any $\phi_{[n]} = (\phi_1,\dots,\phi_n) \in \Disk(n)$,
\begin{align}
\m_{\phi_{[n]}}^{\mathrm{cl}}(F_1,\dots,F_n)= W_{\phi_1}^\cl(F_1)\cdots W_{\phi_n}^\cl(F_n) + \sum \text{contractions of $W_{\phi_1}^\cl(F_1), \dots, W_{\phi_n}^\cl(F_n)$}
\label{eq_conformal_inv}
\end{align}
for $F_i \in P(S^\cl(\bD))$.
\end{itemize}
\end{thm}

\subsection{Harmonic cocycle in two dimensions}\label{sec_central}

In the previous section, we described explicitly the $\Disk$-algebra structure on $P(S^\cl(\bD))$ for $d \geq 3$.
In this section, we consider the case $d=2$, and identify $\R^2$ by $\C$.
We will see that the $\CEt$-algebra structure on $P(S^\cl(\bD))$ can only be obtained after a correction by a harmonic cocycle.

For any injective holomorphic map $\phi:\bD \rightarrow \bD$, set
\begin{align}
H_\phi(z,w) &= \log |\phi(z)-\phi(w)| - \log|z-w| -\ft \log|\phi'(z)| - \ft \log|\phi'(w)|,\\
\label{eq_Green_2d}
\tilde{H}_\phi(z,w) &= \log |\phi(z)-\phi(w)| - \log|z-w|
\end{align}
for $z,w \in \bD \times \bD$ with $z\neq w$. Then, we have:
\begin{prop}\label{prop_Green_2d}
For any $\phi \in \CEt(1)$, the following properties hold:
\begin{enumerate}
\item
$H_\phi(z,w)$ and $\tilde{H}_\phi(z,w)$ are harmonic functions on $\bD \times \bD$.
\item
$H_\phi(z,w)=0$ if and only if $\phi$ is a restriction of an element in $\Conf^+(S^2)$,
that is, there is $\begin{pmatrix}
   a & b \\
   c & d
\end{pmatrix} \in \mathrm{SL}_2\C$ such that $\phi(z) = \frac{az+b}{cz+d}$.
\item
For any $\phi,\psi \in \CEt(1)$,
\begin{align*}
H_{\phi \circ \psi}(z,w)&=H_\psi(z,w)+H_\phi(\psi(z),\psi(w)),\\
\tilde{H}_{\phi \circ \psi}(z,w)&=\tilde{H}_\psi(z,w)+\tilde{H}_\phi(\psi(z),\psi(w))
\end{align*}
for any $z,w \in \bD$.
\end{enumerate}
\end{prop}
\begin{proof}
Since $\phi$ is an injective holomorphic map, $\phi'(z) \neq 0$ for any $z\in \bD$.
Set 
\begin{align*}
F_\phi(z,w)= \frac{\phi'(z)\phi'(w)(z-w)^2}{(\phi(z)-\phi(w))^2},
\end{align*}
which is a holomorphic function on $\bD \times \bD \setminus \{z=w\}$. 
Since $\phi(z)=\phi(w)+(z-w)\phi'(w)+{O}((z-w)^2)$, $F_\phi(z,w)$ defines a holomorphic function on $\bD \times \bD$ and it satisfies
\begin{align}
F_\phi(z,w) = 1+ \frac{1}{6}S(\phi)(z) (w-z)^2+{O}((z-w)^3),\label{eq_Sch}
\end{align}
where
\begin{align}
S(\phi)(z) = \frac{\phi^{'''}(z)}{\phi'(z)} - \frac{3}{2} \left(\frac{\phi^{''}(z)}{\phi'(z)}\right)^2
\label{eq_def_Schwarzian}
\end{align}
 is the Schwarzian derivative.
Since $H_\phi(z,w)=-\ft \log|F_\phi(z,w)|$, (1) follows.
We will show (2).
For $\phi(z) = \frac{az+b}{cz+d}$, it is easy to see that
\begin{align*}
\left(\frac{az+b}{cz+d}-\frac{aw+b}{cw+d}\right) = \frac{1}{(cz+d)(cw+d)}(z-w).
\end{align*}
Hence, $H_\phi=0$. Conversely, assume $H_\phi=0$.
Since $|F_\phi(z,w)|=1$ and $F_\phi$ is holomorphic, $F_\phi(z,w)$
is a constant function. This implies that by \eqref{eq_Sch} the Schwarzian derivative of $\phi$ is zero on $\bD$. Hence, $\phi \in \Conf^+(S^2)$.
(3) follows from $\log|(\phi\circ\psi)'(z)| = 
\log|\psi'(z)(\phi'\circ\psi)(z)|=
\log|\psi'(z)|+\log|(\phi'(\psi(z))|$.
\end{proof}

\begin{rem}\label{rem_Grunsky}
In general, for an injective holomorphic map $\phi:U\rightarrow V$ with $U,V \subset\C$,
\begin{align}
K_\phi(z,w)
=
2\partial_z\partial_wH_\phi(z,w)
=
\left(
\frac{\phi'(z)\phi'(w)}
     {(\phi(z)-\phi(w))^2}
-\frac1{(z-w)^2}
\right)
\label{eq_def_Grunsky}
\end{align}
is the Grunsky kernel, which appears in geometric realizations of Virasoro
representations \cite{KY}.
Its apparent singularity along the diagonal is also removable, and
\begin{align}
K_\phi(z,z)=\frac16 S(\phi)(z)
\label{eq_diagonal_Schwarzian}
\end{align}
where $S(\phi)$ is the Schwarzian derivative in \eqref{eq_def_Schwarzian}.
Note that the same Schwarzian derivative appears in the transformation law of the holomorphic energy-momentum tensor:
\begin{align}
T(z)\mapsto
\phi'(z)^2T(\phi(z))
+\frac{c}{12}S(\phi)(z),
\label{eq_standard_T_transformation}
\end{align}
where $c$ is the central charge
(see, for example, \cite[(4.2.14)]{Zhu}).
%
\end{rem}

%

Let $\phi \in \CEt(1)$ and $f_1,f_2 \in \Cc(\bD)$. 
Now, we consider \eqref{eq_G_d_comp} in the case $d=2$, then
\begin{align}
\begin{split}
&\pa_{G_2}((\Om_\phi(\phi^{-1}(x))^{-2} f_1\circ \phi^{-1}), (\Om_\phi(\phi^{-1}(x))^{-2} f_2\circ \phi^{-1}))\\
&=-(2\pi)^{-1}\int_{\bD \times \bD}\log|\phi(x')-\phi(y')| f_1(x')f_2(y') \mathrm{vol}_{g_\std}(x') \mathrm{vol}_{g_\std}(y')\\
&=-(2\pi)^{-1}\int_{\bD \times \bD}\left(\log|x'-y'|+\tilde{H}_\phi(x',y')\right)f_1(x')f_2(y') \mathrm{vol}_{g_\std}(x') \mathrm{vol}_{g_\std}(y')\\
&=\pa_{G_2-(2\pi)^{-1} \tilde{H}_\phi}(f_1,f_2)
\end{split}
\label{eq_Weyl_res}
\end{align}
Note that $G_2(x,y)-(2\pi)^{-1}\tilde{H}_\phi(x,y)$ satisfies the assumption of Proposition \ref{prop_CG}, since $\tilde{H}_\phi(x,y)$ is a harmonic function on $\bD \times \bD$.
Therefore, as in Proposition \ref{prop_conformal_commute}, the above computation implies that the following diagram commutes:
\begin{align*}
\begin{array}{ccc}
P(\Cc(\bD))/\Delta_{\R^d,g_\std}^{\mathrm{BV}}
&\overset{W_\phi}{\longrightarrow}& P(\Cc(\bD))/\Delta_{\R^d,g_\std}^{\mathrm{BV}}\\
\downarrow_{\Psi_{G_2-(2\pi)^{-1}\tilde{H}_\phi}(\bD)} & & \downarrow_{\Psi_{G_2}(\bD)} \\
P(S^\cl(\bD)) &\overset{W_\phi^{\mathrm{cl}}}{\longrightarrow}& P(S^\cl(\bD)).
\end{array}
\end{align*}
Hence, for any $\phi \in \Disk(1)$,  $\m_\phi^{\mathrm{cl}}$ in \eqref{def_pullback_algebra} is equal to 
 \begin{align}
W_\phi^{\mathrm{cl}}\circ \exp(-(2\pi)^{-1}\pa_{\tilde{H}_\phi}):P(S^\cl(\bD)) \rightarrow P(S^\cl(\bD)),
\label{eq_add_formula_rho_cl}
\end{align}
which gives a monoid homomorphism.

Thus, in two dimensions Proposition \ref{prop_conformal_commute} does not hold, and a correction by a harmonic cocycle is required.
Because of the property in Proposition \ref{prop_Green_2d} (2), it is preferable to use $H_\phi$ rather than $\tilde{H}_\phi$.
Rewriting \eqref{eq_Weyl_res} in terms of $H_\phi$, we obtain
\begin{align}
\begin{split}
&\pa_{G_2}((\Om_\phi(\phi^{-1}(x))^{-2} f_1\circ \phi^{-1}), (\Om_\phi(\phi^{-1}(x))^{-2} f_2\circ \phi^{-1}))\\
&=-(2\pi)^{-1}\int_{\bD \times \bD}\left(\log|x'-y'|+ \ft \log |\phi'(x')\phi'(y')| +H_\phi(x',y')\right)f_1(x')f_2(y') \mathrm{vol}_{g_\std}(x') \mathrm{vol}_{g_\std}(y')\\
&=
\pa_{G_2-(2\pi)^{-1} H_\phi}(f_1,f_2)-\frac{1}{4\pi} \int_{\bD \times \bD}\left(\log |\phi'(x')|+\log |\phi'(y')|\right)f_1(x')f_2(y') \mathrm{vol}_{g_\std}(x') \mathrm{vol}_{g_\std}(y').
\end{split}
\label{eq_Weyl_res2}
\end{align}
Set
\begin{align*}
\Cc(\bD)_0 =\left\{h \in \Cc(\bD) \mid \int_{\bD} h d^2x=0 \right\}.
\end{align*}
If $f_1,f_2 \in \Cc(\bD)_0$, then
\begin{align*}
\int_{\bD \times \bD}&\log |\phi'(y')| f_1(x')f_2(y') \mathrm{vol}_{g_\std}(x') \mathrm{vol}_{g_\std}(y')\\
&=\int_\bD f_1(x')  \mathrm{vol}_{g_\std}(x') \int_{\bD}\log |\phi'(y')| f_2(y')\mathrm{vol}_{g_\std}(y')=0
\end{align*}
Therefore, in this case,
\begin{align}
\pa_{G_2}(W_\phi^\cl(f_1), W_\phi^\cl(f_2))=
\pa_{G_2-(2\pi)^{-1}H_\phi}(f_1,f_2)
\end{align}
holds. 
Note that if $h \in \Cc(\bD)$, then $\int \Delta_x h(x)d^2x =0$. Hence,
\begin{align*}
S_0^\cl(\bD) = \Cc(\bD)_0 / \Delta \Cc(\bD)
\end{align*}
is a codimension $1$ subspace of $S^\cl(\bD)$.
Let $[f] \in S_0^\cl(\bD)$. Then, for any $\phi \in \CEt(1)$,
\begin{align}
\int_{\bD} W_\phi^\cl(f) d^2x &= \int_{\phi(\bD)} \Om_\phi(\phi^{-1}(x))^{-2} f\circ \phi^{-1}(x) d^2x =\int_{\bD} \Om_\phi(y)^{-2} f(y) \Om_\phi(y)^2 d^2y =0,
\label{eq_push_zero}
\end{align}
where $y=\phi^{-1}(x)$.
Thus, $W_\phi^\cl$ preserves $P(S_0^\cl(\bD))$. Similarly, by restricting $\Delta_{\R^d,g_\std}^{\mathrm{BV}}$ to $\Cc(\bD)_0$, we obtain
\begin{align*}
\Delta_{\R^d,g_\std}^{\mathrm{BV}}:P(\Cc(\bD)_0)\otimes \Cc(\bD) \rightarrow P(\Cc(\bD)_0)
\end{align*}
\begin{prop}\label{prop_conformal_commute2}For $\phi \in \CEt(1)$, the following diagram commutes:
\begin{align*}
\begin{array}{ccc}
P(\Cc(\bD)_0)/\Delta_{\R^d,g_\std}^{\mathrm{BV}}
&\overset{W_\phi}{\longrightarrow}& P(\Cc(\bD)_0)/\Delta_{\R^d,g_\std}^{\mathrm{BV}}\\
\downarrow_{\Psi_{G_2-(2\pi)^{-1}H_\phi}(\bD)} & & \downarrow_{\Psi_{G_2}(\bD)} \\
P(S_0^\cl(\bD)) &\overset{W_\phi^{\mathrm{cl}}}{\longrightarrow}& P(S_0^\cl(\bD)).
\end{array}
\end{align*}
\end{prop}

By \eqref{eq_push_zero}, $P(\Cc(\bD)_0)/\Delta_{\R^d,g_\std}^{\mathrm{BV}}$ is a subalgebra of the $\CEt$-algebra $\HY(\bD)$. We denote it by $\HY(\bD)_0$.
Hence, by \eqref{def_pullback_algebra}, $P(S_0^\cl(\bD))$ inherits a $\CEt$-algebra structure.
As in Theorem \ref{thm_classical_disk}, we have:
\begin{thm}\label{thm_classical_disk2}
The $\CEt$-algebra $P(S_0^\cl(\bD))$
satisfies
\begin{itemize}
\item
For any $\phi \in \CEt(1)$,
 $\m_\phi^{\mathrm{cl}}$ is equal to 
 \begin{align*}
W_\phi^{\mathrm{cl}}\circ \exp(-(2\pi)^{-1}\pa_{H_\phi}):P(S_0^\cl(\bD)) \rightarrow P(S_0^\cl(\bD)).
\end{align*}
\item
 For any $\phi_{[n]} = (\phi_1,\dots,\phi_n) \in \CEt(n)$ and $F_i \in P(S_0^\cl(\bD))$,
\begin{align*}
\m_{\phi_{[n]}}^{\mathrm{cl}}(F_1,\dots,F_n)&= W_{\phi_1}^\cl(\exp(-(2\pi)^{-1}\pa_{H_{\phi_1}}) F_1)\cdots W_{\phi_n}^\cl(\exp(-(2\pi)^{-1}\pa_{H_{\phi_n}}) F_n)\nonumber \\
&+ \sum \text{contractions of $W_{\phi_1}^\cl(\exp(-(2\pi)^{-1}\pa_{H_{\phi_1}}) F_1), \dots, W_{\phi_n}^\cl(\exp(-(2\pi)^{-1}\pa_{H_{\phi_n}})F_n)$}.
\end{align*}
\end{itemize}
\end{thm}

\begin{rem}\label{rem_extend_disk2}
For later use, we extend \eqref{eq_add_formula_rho_cl} to a slightly more general setting.
Let $U,V\subset \R^2$ be open subsets and let
$\phi:U\rightarrow V$ be an injective holomorphic map.
Define
\[
S^\cl(\phi):P(S^\cl(U))\longrightarrow P(S^\cl(V))
\]
by the composition
\begin{align}
S^\cl(\phi):
P(S^\cl(U))
\overset{\Psi_G(U)^{-1}}{\longrightarrow}
\HY(U)
\overset{\HY(\phi)}{\longrightarrow}
\HY(V)
\overset{\Psi_G(V)}{\longrightarrow}
P(S^\cl(V)).
\end{align}
As in \eqref{eq_def_W_add2}, $W_\phi$  induces a linear map $W_\phi^\cl:P(S^\cl(U))\rightarrow P(S^\cl(V))$.
The function $\tilde{H}_\phi$ in \eqref{eq_Green_2d} is harmonic
on $U\times U$, and the computation in \eqref{eq_Weyl_res}
gives
$\partial_{G_2}
\bigl(W_\phi^\cl(f_1),W_\phi^\cl(f_2)\bigr)
=
\partial_{G_2-(2\pi)^{-1}\tilde{H}_\phi}(f_1,f_2)$
for $f_1,f_2\in C_c^\infty(U)$.
Therefore, by the same argument as for \eqref{eq_add_formula_rho_cl},
\begin{align}
S^\cl(\phi)
=
W_\phi^\cl
\circ
\exp\bigl(-(2\pi)^{-1}\pa_{\tilde H_\phi}\bigr).
\label{eq_add_general_formula}
\end{align}
In particular, if $U\subset V$ and $\phi$ is the inclusion map, then
$\tilde H_\phi=0$, and hence
$S^\cl(\phi)=W_\phi^\cl$.
\end{rem}

We end this section with a remark on the conformal anomaly.
Let $\pr_0:
P(S^\cl(\bD))= \bigoplus_{n=0}^\infty \Sym^n(S^\cl(\bD)) \rightarrow \R$ be the projection onto the constant term of the polynomial algebra.
Also, let $\langle \bullet \rangle_{P}:\HY(\bD) \rightarrow \R$ be the linear map given by
\begin{align}
\langle \bullet \rangle_{P}:
\HY(\bD) \overset{\Psi_P(\bD)}{\rightarrow} P(S^\cl(\bD))\overset{\pr_0}{\rightarrow} \R
\label{eq_def_state}
\end{align}
for any $P$ in Proposition \ref{prop_CG}.

The map $\langle \bullet \rangle_P$ will be called a state, in analogy with the notion of a state for a prefactorization algebra in \cite[Definition 4.9.1]{CG}. It allows us to define physical quantities such as correlation functions.
To specify such a state, we need to choose a boundary condition for the fundamental solution $P$ of the Laplacian as in Proposition \ref{prop_CG}
(cf., (3) in Remark \ref{rem_Green_characterize}).
When the boundary condition can be chosen in a conformally invariant way, the resulting state is conformally invariant,
which is not the case in two dimensions.
\begin{prop}\label{prop_Weyl_anomaly}
If $d \geq 3$, then for any $\phi \in \Disk(1)$ and $f_1 \dots f_n \in \Cc(\bD)$,
\begin{align*}
\langle f_1 \dots f_n \rangle_{G_d} =\langle \m_\phi^\cl (f_1) \dots \m_\phi^\cl(f_n) \rangle_{G_d}.
\end{align*}
If $d = 2$, then for any $\phi \in \CEt(1)$ and $f_1 \dots f_n \in \Cc(\bD)$,
\begin{align*}
\langle f_1 \dots f_n \rangle_{G_2-(2\pi)^{-1}\tilde{H}_\phi} =\langle \m_\phi^\cl (f_1) \dots \m_\phi^\cl(f_n) \rangle_{G_2}.
\end{align*}
Moreover, if $\phi \in \Conf^+(S^2)$ and $f_1 \dots f_n \in \Cc(\bD)_0$, then 
\begin{align*}
\langle f_1 \dots f_n \rangle_{G_2} =\langle \m_\phi^\cl (f_1) \dots \m_\phi^\cl(f_n) \rangle_{G_2}.
\end{align*}
\end{prop}
\begin{proof}
By Proposition \ref{prop_conformal_commute2}, we have
\begin{align*}
\langle \m_\phi^\cl (f_1) \dots \m_\phi^\cl(f_n) \rangle_{G_2}&=\pr_0 \exp(\pa_{G_2}) \m_\phi^\cl (f_1) \dots \m_\phi^\cl(f_n)\\
&=\pr_0  W_\phi^\cl \exp(\pa_{G_2}-(2\pi)^{-1}\pa_{\tilde{H}_\phi}) (f_1 \dots f_n)\\
&=\pr_0  \exp(\pa_{G_2}-(2\pi)^{-1}\pa_{\tilde{H}_\phi}) (f_1 \dots f_n)=\langle f_1 \dots f_n \rangle_{G_2-(2\pi)^{-1}\tilde{H}_\phi}.
\end{align*}
\end{proof}

\section{Observable on $\R^d$ and harmonic functions}\label{sec_obs}
In the previous section we showed that $\HY$ defines a symmetric monoidal functor from $\Emb$ to $\Vect$, and we studied its restriction to $(\R^d,g_\std)$. In particular, via the linear isomorphism
\begin{align*}
\Psi_{G_d}(\bD):\HY(\bD) \overset{\cong}{\rightarrow} P(S^\cl(\bD))
\end{align*}
the space $P(S^\cl(\bD))$ carries the structure of a $\Disk$-algebra. In this section we give an explicit description of
\begin{align*}
S^\cl(U) = \mathrm{cok}(\Delta:\Cc(U) \rightarrow \Cc(U))\qquad (U \subset \R^d)
\end{align*}
in terms of harmonic functions. In Section \ref{sec_compact} we show that $S^\cl(U)$ agrees with $H'(U)$ (the topological dual of the space of harmonic functions on $U$). In Section \ref{sec_expansion}, we give an explicit description of $H'(\bD)$ using the expansion in harmonic polynomials, and show that the space of observables of the prefactorization algebra is embedded into the Hilbert Fock space $\hat{\mathrm{Sym}}(H_\CFT)$, which plays an important role in the axiomatic QFT (see for example \cite{RS2, Arai}).
In Section \ref{sec_algebra}, we explicitly describe the $\Disk$-algebra structure on $P(H'(\bD))$.
Finally, in Section \ref{sec_Ward}, we prove the conformal Ward identity and relate the harmonic cocycle introduced in Section~\ref{sec_central} to the Gelfand--Fuchs cocycle.

\subsection{Classical observable as distributions on harmonic functions}\label{sec_compact}
Let $d\geq 2$ and $U \subset \R^d$ be an open subset.
In this section, we will describe $\mathrm{cok}(\Delta:C_c^\infty(U) \rightarrow C_c^\infty(U))$.
Set $\Delta = - \sum_{i=1}^d \pa_i^2$ and 
\begin{align*}
H(U) = \{f\in C^\infty(U) \mid \Delta f =0\},
\end{align*}
the space of harmonic functions, which is a topological vector space as a subspace of $C^\infty(U)$,
namely, the topology is defined by the family of seminorms
\begin{align}
p_{n,m}(u) = \max_{|\al|\leq m} \sup_{x\in K_n} |\partial^\al u(x)|\label{def_top_harm}
\end{align}
for $u \in H(U)$, where $\{K_n \subset U\}_{n\in \N}$ is a sequence of compact subsets such that $U= \cup_n K_n$ and $K_n \subset K_{n+1}^\circ$ (see Appendix \ref{sec_app}).
Since $\Delta:C^\infty(U) \rightarrow C^\infty(U)$ is continuous, $H(U) \subset C^\infty(U)$ is a closed subspace.
Denote by $H'(U)$ the topological dual of $H(U)$.

Let $B:\Cc(U) \otimes H(U) \rightarrow \R$
be a bilinear form defined by
\begin{align*}
B(f,u) = \int_U f(x)u(x) d^dx
\end{align*}
for any $f \in \Cc(U)$ and $u\in H(U)$.
Since $|\int_U f(x)u(x) dx^d| \leq (\int_U|f(x)|dx) \sup_{x\in \supp(f)}|u(x)|$,
for any $f \in \Cc(U)$, the linear map $B(f,\bullet): H(U) \rightarrow \R$ is continuous.
Hence, we have a linear map $\Cc(U) \rightarrow H'(U)$.
Moreover, if $g \in \Cc(U)$, then $B(\Delta g,\bullet)=0$ since for any $u \in H(U)$,
\begin{align*}
B(\Delta g, u) = \int_U (\Delta g)(x) u(x) d^dx =\int_U g(x) (\Delta u)(x) d^dx =0.
\end{align*}
Hence, $B$ induces a linear map
\begin{align*}
T_U: \Cc(U) / \Delta \Cc(U) \rightarrow H'(U).
\end{align*}

\begin{thm}\label{thm_compact_isomorphism}
Let $d \geq 2$ and $U \subset \R^d$ be an open subset. Then, 
$T_U: \Cc(U) / \Delta \Cc(U) \rightarrow H'(U)$ is a linear isomorphism.
\end{thm}
To prove this theorem, we need some basic properties of distributions, which are collected in the appendix.
\begin{lem}\label{lem_cpt_surj}
$T_U$ is surjective.
\end{lem}
\begin{proof}
Let $T \in H'(U)$. Then, there is a compact subset $K \subset U$ and $C>0$, $n \in \N$ such that
\begin{align*}
|T(u)| \leq C p_{K,n}(u) = C \max_{|\al|\leq n} \sup_{x \in K}|\partial^{\al}u(x)|
\end{align*}
for any $u \in H(U)$.
Then, by the Hahn-Banach theorem, there is $\tT \in C^\infty(U)'$ such that
$\tT|_{H(U)} =T$ and $|\tT(f)| \leq C p_{K,n}(f)$ for any $f\in C^\infty(U)$.
Then, $\supp(\tT) \subset K \subset U$.

Choose $\ep>0$ sufficiently small so that $K + \overline{B_\ep(0)} \subset U$,
where $B_\ep(0)=\{x\in \R^d \mid \norm{x}<\ep\}$.
Let $\chi$ be a smooth function supported in $\overline{B_\ep(0)}$ and equal to $1$ in a neighborhood of the origin, and let 
\begin{align}
\begin{split}
\Psi(x)=
\begin{cases}
\frac{1}{(d-2)\om_d \norm{x}^{\,d-2}} & d\geq 3\\
-(2\pi)^{-1} \log |x| & d=2
\end{cases}
\end{split}
\label{eq_fundamental_sec2}
\end{align}
be the fundamental solution. Then, $\Psi(x)$ is a locally-integrable function, and thus, $\Psi(x) \in \Cc(\R^d)'$, and $(\chi \Psi)(x)$ is a compactly supported distribution in $\R^d$.
Set $R(x) = \Delta (\chi \Psi) - \delta_0 = \Delta ((\chi -1)\Psi) \in C^\infty(\R^d)'$.
Since $(\chi -1)\Psi \in C^\infty(\R^d)$ and $\pa_i \Psi(x)$ is locally integrable \cite[Theorem 3.3.2]{Hormander},
\begin{align*}
R(x) =\Delta (\chi) \Psi - 2\sum_j (\partial_j \chi)(\partial_j \Psi).
\end{align*}
Then, $\mathrm{supp}R \subset \overline{B_\ep(0)}$ and $R(x)$ is smooth, since $\partial_j \chi =0$ around the origin.
Since $\tT$ is a compactly supported distribution,
we can regard it as an element of $C^\infty(\R^d)'$. Then, the convolution satisfies
\begin{align}
\Delta(\chi \Psi) * \tT = (\delta_0+R(x)) * \tT = \tT+R * \tT.
\label{eq_C_infty1}
\end{align}
Since $\supp(R*\tT), \supp(\Delta(\chi \Psi) * \tT) \subset K + \overline{B_\ep(0)} \subset U$ by Proposition \ref{prop_conv_distribution},
$\Delta(\chi \Psi) * \tT$ and $R * \tT$ can be regarded as elements in $C^\infty(U)'$.
Namely, let $\eta \in \Cc(U)$ be a function that is identically $1$ on a neighborhood of $K+\overline{B_\ep(0)}$.
Then, for $f \in C^\infty(U)$, the pairing $\langle \Delta(\chi \Psi) * \tT, f \rangle$ is defined by
$\langle \Delta(\chi \Psi) * \tT, \eta f \rangle$. Here $\eta f$ is regarded as an element of $\Cc(\R^d)$ via zero extension (see Proposition \ref{app_dist_compact}).
Then, for any $u \in H(U)$, by Proposition \ref{prop_conv_distribution}, $\langle \Delta(\chi \Psi) * \tT,u  \rangle = \langle \Delta(\chi \Psi * \tT),\eta u  \rangle =\langle \chi \Psi * \tT, \Delta (\eta u)  \rangle$.
Since $\Delta(\eta u)$ vanishes on a neighborhood of $K+\overline{B_\ep(0)}$, 
$\langle \Delta(\chi \Psi) * \tT,u  \rangle = 0$. Hence, we have
\begin{align*}
\langle T,u  \rangle&=\langle \tilde{T},u  \rangle=\langle \Delta(\chi \Psi) * \tT - R* \tT,u  \rangle= - \langle R* \tilde{T},u  \rangle.
\end{align*}
Since $R * \tT \in \Cc(U)$ by Proposition \ref{prop_convolutions}, $T(u) =  - \int  (R *\tT) u d^dx = B(-R*\tT,u)$.
\end{proof}

\begin{lem}\label{lem_cpt_rad}
Assume that $f \in \Cc(U)$ satisfies $B(f,\phi)=0$ for any $\phi \in H(U)$.
Then, there is $h\in \Cc(U)$ such that $f = \Delta h$.
In other words, $T_U$ is injective.
\end{lem}
\begin{proof}
Consider the convolution of $f$ with the fundamental solution $\Psi(x) \in \Cc(\R^d)'$ in \eqref{eq_fundamental_sec2}.
By Proposition \ref{prop_convolutions}, $\Psi * f \in C^\infty(\R^d)$ and satisfies $\Delta(\Psi * f) = f$.
Hence, it suffices to show that $G * f$ has compact support in $U$, where $G(x)=\frac{1}{\norm{x}^{d-2}}$ ($d \geq 3$) and $G(x) = 2\log |x|$ ($d=2$).
Set $K=\mathrm{supp}f \subset U$.
For $x_0\notin U$,
\begin{align*}
(G*f)(x_0)&=\int_{\R^d} G(x_0-y)f(y)d^dy=\int_{K} G(x_0-y)f(y)d^dy.
\end{align*}
Since the restriction of $G(x_0-y)$ on $y \in U$ is harmonic, $(G*f)(x_0)=0$ by the assumption.
Hence, $\supp(G*f) \subset \overline{U}$.

Since $K$ is compact, there is $R>0$ such that $K \subset B_R(0)$.
We claim that $(G*f)(x_0)=0$ for any $\norm{x_0}>R$.
In fact, for $d \geq 3$, by \eqref{eq_Harm_expansion}, $G$ admits an expansion in harmonic polynomials which converges uniformly on $K$ (see also Proposition \ref{prop_regular_harmonic}).
 Hence, we have:
\begin{align}
\begin{split}
(G*f)(x_0)&=\int_{K} G(x_0-y)f(y)d^dy\\
&=\int_{K} \sum_{n \geq 0} \left(\sum_{k=1}^{\dim \Harm_{d,n}}\frac{c_{d,n}}{\norm{x_0}^{2n+d-2}}Y_{n,k}(x_0)Y_{n,k}(y)\right)f(y)d^dy\\
&=\sum_{n \geq 0}\int_{K}  \left(\sum_{k=1}^{\dim \Harm_{d,n}}\frac{c_{d,n}}{\norm{x_0}^{2n+d-2}}Y_{n,k}(x_0)Y_{n,k}(y)\right)f(y)d^dy\\
&=\sum_{n \geq 0}\left(\sum_{k=1}^{\dim \Harm_{d,n}}\frac{c_{d,n}}{\norm{x_0}^{2n+d-2}}Y_{n,k}(x_0) \int_{K} Y_{n,k}(y)f(y)d^dy \right).
\end{split}
\label{eq_GfR_d3}
\end{align}
Since the harmonic polynomial $Y_{n,k}(y)$ can be extended onto $\R^d$ as a harmonic function, and thus, $Y_{n,k}(y) \in H(U)$.
Hence, $(G*f)(x_0)=0$ for any $\norm{x_0}>R$ and $\supp(G*f) \subset \overline{B_R(0)} \cap \overline{U}$ and $G*f$ is compactly supported.
Similarly, in the case $d=2$, \eqref{eq_GfR_d3} becomes, for $z_0 \in \C$ satisfying $|z_0|>R$,
\begin{align*}
(G*f)(z_0)&=\int_{K} 2\log|z_0(1-w/z_0)|f(w)d^2w\\
&=\int_{K}\left(2 \log|z_0| - \sum_{n \geq 1}\frac{1}{n}((w/z_0)^n+(\bar{w}/\z_0)^n)\right)f(w)d^2w =0.
\end{align*}
Thus, it suffices to prove that $\supp(G*f) \subset U$.
Let $x_0 \in \overline{U} \setminus U$
and 
\begin{align*}
\rho = \inf_{y \in K} \norm{x_0-y}>0.
\end{align*}
Then, in the case of $d \geq 3$, for any $x \in B_{\rho}(x_0)$,
\begin{align*}
(G*f)(x)&=\int_{\R^d} G(x-y)f(y)d^dy=\int_{K} G((x-x_0)-(y-x_0))f(y)d^dy\\
&= \int_K 
\sum_{n \geq 0} \left(\sum_{k=1}^{\dim\Harm_{d,n}}\frac{c_{d,n}}{\norm{y-x_0}^{2n+d-2}}Y_{n,k}(y-x_0)Y_{n,k}(x-x_0)\right)f(y)d^dy\\
&= \sum_{n \geq 0} \left(\sum_{k=1}^{\dim\Harm_{d,n}} Y_{n,k}(x-x_0) \int_K 
\frac{c_{d,n}}{\norm{y-x_0}^{2n+d-2}}Y_{n,k}(y-x_0)f(y)d^dy \right).
\end{align*}
Since $\frac{1}{\norm{y-x_0}^{2n+d-2}}Y_{n,k}(y-x_0)$ is a smooth harmonic function except for $y=x_0$ (see for example, \cite[Theorem 5.18]{ABR}),
$\frac{1}{\norm{y-x_0}^{2n+d-2}}Y_{n,k}(y-x_0) \in H(U)$.
Hence, $(G*f)(x)=0$ for any $x \in B_{\rho}(x_0)$, and thus, $G*f \in \Cc(U)$.
Similarly, in the case of $d =2$, for any $z \in \overline{U} \setminus U$, $z \in B_{\rho}(z_0)$,
\begin{align*}
&(G*f)(z)=\int_{K} 2\log|(z-z_0)-(w-z_0)|f(w)d^2w\\
&= \int_K \left( 2\log|w-z_0| - \sum_{n \geq 1}\frac{1}{n}\left(((z-z_0)/(w-z_0))^n + ((\z-\z_0)/(\w-\z_0))^n \right)\right)f(w)d^2w=0.
\end{align*}
Hence, the assertion holds.
\end{proof}
By Lemma \ref{lem_cpt_surj} and Lemma \ref{lem_cpt_rad}, Theorem \ref{thm_compact_isomorphism} is obtained.

Next, we extend Theorem \ref{thm_compact_isomorphism} to the case of $\Cc(U)_0/\Delta \Cc(U)$ for $d=2$.
Recall
\begin{align*}
\Cc(U)_0=\{f \in \Cc(U) \mid \int f d^dx =0\}.
\end{align*}
Note that $\Cc(U)_0/\Delta \Cc(U)$ is a codimension one subspace of $\Cc(U)/\Delta \Cc(U)$.
Set
\begin{align*}
H'(U)_0 = \{T \in H'(U) \mid T(1_U)=0 \},
\end{align*}
where $1_U$ is the constant function on $U$. Then, for any $f \in \Cc(U)_0$, by definition we have
\begin{align*}
\langle T_U(f), 1_U \rangle = \int_U f d^dx =0
\end{align*}
and hence the image of $\Cc(U)_0/\Delta \Cc(U)$ in $H'(U)$ under $T_U$ is $H'(U)_0$.
\begin{cor}\label{cor_restriction_T}
The restriction of $T_U$ gives a linear isomorphism $\Cc(U)_0/\Delta \Cc(U) \rightarrow H'(U)_0$.
\end{cor}

Let $\al \in \N^d$ and 
\begin{align*}
\partial^\al \delta_0: H(\bD) \rightarrow \R,\qquad
u \mapsto (-1)^{|\al|}(\partial^\al u)(0),
\end{align*}
which defines a continuous linear map.
Since $H(\bD)' \cong \Cc(\bD)/\Delta \Cc(\bD)$, there is a compactly supported function $f$ which represents $\partial^\al \delta_0$.
We end this section by explicitly giving this function.

\begin{lem}\label{lem_dirac}
Let $U \subset \R^d$ be an open subset and $a\in U$.
Let $\ep >0$ with $\overline{B_\ep(a)} \subset U$.
Assume that $f_\ep^a \in \Cc(U)$ satisfies
\begin{enumerate}
\item
$\supp(f_\ep^a) \subset \overline{B_\ep(a)}$.
\item
$f_\ep^a(R x+a) = f_\ep^a(x+a)$ for any $R \in \mathrm{SO}(d)$ and $x \in \overline{B_\ep(0)}$.
\item
$\int_U f_\ep^a(x) d^dx =1$.
\end{enumerate}
Then, $f_\ep^a(x)$ represents $\delta_a$, that is, for any $u \in H(U)$,
\begin{align*}
u(a) = \int_{U} f_\ep^a(x) u(x) d^dx.
\end{align*}
Moreover, $\pa^\al \delta_a$ is represented by $\pa^\al f_\ep^a(x)$ for any $\al \in \N^d$.
\end{lem}
\begin{proof}
Since $u$ is harmonic on $U$ and $\overline{B_\ep(a)} \subset U$,
by the mean value property of harmonic functions \cite[1.4]{ABR},
\begin{align*}
u(a) =\int_{S^{d-1}} u(r\om+a) \frac{d\om}{\om_d}
\end{align*}
for any $0 < r \leq \ep$, where $\om_d$ is the volume of $S^{d-1}$. Hence, using spherical coordinates,
\begin{align*}
 \int_{U} f_\ep^a(x) u(x) d^dx&=
\int_{\overline{B_\ep(a)}} f_\ep^a(x) u(x)dx^d\\
&=\int_{\overline{B_\ep(0)}} f_\ep^a(x+a) u(x+a)dx^d\\
&=\int_{r=0}^\ep \int_{S^{d-1}}  f_\ep^a(r\om+a) u(r\om+a) r^{d-1} d\om dr\\
&=\int_{r=0}^\ep f_\ep^a(r e_1+a)r^{d-1}\int_{S^{d-1}} u(r\om+a) d\om dr\\
&=u(a) \int_{r=0}^\ep f_\ep^a(r e_1+a)\om_d r^{d-1} dr.
\end{align*}
Hence, the assertion follows from
\begin{align*}
\int_{r=0}^\ep f_\ep^a(r e_1+a)\om_d r^{d-1} dr
&= \int_{r=0}^\ep\int_{S^{d-1}} f_\ep^a(r \om+a)r^{d-1}d\om dr\\
&= \int_{B_\ep(0)} f_\ep^a(x+a)d^dx\\
&=\int_U f_\ep^a(x) d^dx =1.
\end{align*}
\end{proof}

Recall the following classical result:
\begin{prop}\cite[Theorem 2.3.4]{Hormander}\label{prop_finite_dist}
Let $x_0 \in U$ and $T:C^\infty(U) \rightarrow \R$ be a distribution supported at $\{x_0\}$.
Then, there is $N>0$ and $a_\al \in \R$ such that
\begin{align}
T= \sum_{|\al|\leq N} a_\al \pa^\al \delta_{x_0}.\label{eq_finite_support_classical}
\end{align}
\end{prop}
Therefore, a distribution supported at $\{0\}$ can only be a finite sum as in \eqref{eq_finite_support_classical}.
However, using the real-analyticity of harmonic functions, one sees that $H'(\bD)$ contains infinite sums of derivatives of the Dirac delta function. In the next section, we will study this expansion of harmonic distributions.

\subsection{Growth conditions of distributions on harmonic functions}\label{sec_expansion}
Let $d \geq 2$.
In this section, using basic properties of harmonic functions defined on $B_r(0)$, we give an explicit description of $H'(B_r(0))$.
The contents of this section are an analogue of the harmonic analysis of $C^\infty(S^{d-1})$ in \cite{EK}.

First, we recall the well-known fact that the $C^\infty$-topology on $H(B_r(0))$ (see Definition \ref{def_topology}) agrees with the topology of uniform convergence on compact sets. 
\begin{prop}\cite[Theorem 2.4]{ABR}
\label{prop_analytic_harm}
Let $r>0$ and $a\in \R^d$.
For any $\al \in \N^d$, there is a constant $C_\al >0$ such that
\begin{align*}
|(\partial^\al u)(a)|\leq \frac{C_\al}{r^{|\al|}}M
\end{align*}
for any harmonic function $u$ on $B_r(a)$ which is bounded by $M>0$ on ${\overline{B_r(a)}}$.
\end{prop}
The following is clear from the above proposition:
\begin{cor}
Define a seminorm $p_s$ ($0<s<r$) on $H(B_r(0))$ by
\begin{align}
p_s(u) = \sup_{x\in B_s(0)}|u(x)|\label{eq_app_seminorm_uniform}
\end{align}
The topology on $H(B_r(0))$ defined by $\{p_s\}$ (the topology of uniform convergence on compact sets) is equivalent to the topology defined in Definition \eqref{def_top_harm}.
\end{cor}

For $n \geq 0$, set
\begin{align*}
\Harm_{d,n} = \{P(x) \in \R[x_1,\dots,x_d]\mid \Delta P(x)=0 \text{ and the total degree of }P(x) \text{ is }n \},
\end{align*}
the space of harmonic polynomials of degree $n$.
\begin{prop}\cite[Corollary 5.34]{ABR}
\label{prop_regular_harmonic}
Let $u$ be a harmonic function on $B_r(0)$. Then, there are unique harmonic polynomials $u_n(x) \in \Harm_{d,n}$ such that
\begin{align*}
u(x) = \sum_{n\geq 0} u_n(x)
\end{align*}
on $B_r(0)$, the series converging absolutely and uniformly on compact subsets of $B_r(0)$.
\end{prop}
\begin{rem}
Although harmonic functions are real-analytic,
Proposition~\ref{prop_regular_harmonic} does not imply that the corresponding
power series in the real variables converges absolutely
throughout $B_r(0)$.
Indeed, $1/(1-z)=1/(1-x-iy) =\sum_{n\geq 0}\sum_{k=0}^n \binom{n}{k}x^k (iy)^{n-k}$
is a harmonic function on $\{|z|<1\}$, but it converges absolutely only when $|x|+|y|<1$.
\end{rem}
By the above proposition, the space of harmonic functions embeds into the space of sequences of harmonic polynomials:
\begin{align}
\T_H:H(B_r(0)) \rightarrow \prod_{n \geq 0}\Harm_{d,n},\qquad u \mapsto (u_n)_{n\geq 0}\label{eq_app_exp_harm}
\end{align}
The purpose of this section is to give characterizations of spaces such as $H(B_r(0))$ and its topological dual in terms of $(u_n)_{n\geq 0} \in \prod_{n \geq 0}\Harm_{d,n}$.

Set $\Harm_d = \bigoplus_{n \geq 0}\Harm_{d,n}$.
By restricting a harmonic polynomial $u \in \R[x_1,\dots,x_d]$ to $S^{d-1}$, we obtain the following inner product $(-,-)_{L^2(S^{d-1})}: \Harm_{d} \otimes \Harm_{d} \rightarrow \R$:
%
\begin{align}
(u,v)_{L^2(S^{d-1})}= \int_{S^{d-1}} {u(\xi)}{v(\xi)} d\si(\xi)\label{eq_def_L2}
\end{align}
for $u,v \in \Harm_{d}$. Here we normalize $d\si(\xi)$ so that it satisfies $\int_{S^{d-1}} d\si(\xi)=1$.
The completion of $\Harm_d$ with respect to $(-,-)_{L^2(S^{d-1})}$ is $L^2(S^{d-1},d\si)$, and
\begin{align*}
L^2(S^{d-1},d\si) = \hat{\bigoplus_{n \geq 0}} \Harm_{d,n}
\end{align*}
gives an orthogonal decomposition of the Hilbert space (see \cite[Theorem 5.12]{ABR}).
%

\begin{lem}\label{lem_app_compare_norm}
Set $A_{d,n} = \dim_\R \Harm_{d,n}$, which is a degree $d-2$ polynomial of $n$. For any $u_n \in \Harm_{d,n}$,
\begin{align*}
\mathrm{sup}_{\xi \in S^{d-1}} |u_n(\xi)| &\leq \sqrt{A_{d,n}} \norm{u_n}_{L^2(S^{d-1})}\\
\norm{u_n}_{L^2(S^{d-1})} &\leq \mathrm{sup}_{\xi \in S^{d-1}}|u_n(\xi)|.
\end{align*}
\end{lem}
\begin{proof}
The first inequality can be found in \cite[Proposition 2.6]{MLeft}. The second follows from
\begin{align*}
\norm{u_n}_{L^2(S^{d-1})}^2= \int_{S^{d-1}} |u_n(\xi)|^2 d\si(\xi)\leq  (\mathrm{sup}_{\xi \in S^{d-1}}|u_n(\xi)|)^2.
\end{align*}
\end{proof}
We will denote $\mathrm{sup}_{\xi \in S^{d-1}}|u(\xi)|$ by $\norm{u}_{S^{d-1},\infty}$.
Define a family of seminorms $q_s$ ($s >0$) on the sequence space $(u_n)_n \in \prod_{n \geq 0} \Harm_{d,n}$  by
\begin{align*}
q_s((u_n)_n) = \sum_{n\geq 0} s^n \norm{u_n}_{L^2(S^{d-1})}.
\end{align*}

\begin{prop}\label{prop_equivalence}
For a sequence $(u_n)_n \in \prod_{n \geq 0} \Harm_{d,n}$, the following conditions are equivalent:
\begin{enumerate}
\item
$\sum_{n \geq 0} u_n(x)$ is an expansion of a harmonic function on $B_r(0)$;
\item
$q_s((u_n)_n)<\infty$ for any $s$ with $r>s>0$.
\end{enumerate}
Moreover, the family of seminorms $\{q_s\}_{r>s>0}$ is equivalent to \eqref{eq_app_seminorm_uniform}.
\end{prop}
\begin{proof}
Assume (2) holds. Then, by $\sup_{|x| \leq s}|u_n(x)| =s^n \sup_{|x| =1}|u_n(x)|$, for any $t<r$,
\begin{align}
\sum_{n \geq 0} \sup_{|x| \leq t}|u_n(x)| &=\sum_{n \geq 0} t^n \norm{u_n(x)}_{S^{d-1},\infty}\leq \sum_{n \geq 0} t^n \sqrt{A_{d,n}} \norm{u_n(x)}_{L^2(S^{d-1})}.
\label{eq_seminorm_equiv}
\end{align}
Since $\sqrt{A_{d,n}}$ has polynomial growth in $n$, the condition $q_s(u)<\infty$ for $t<s<r$ implies that
$\sum_{n \geq 0} \sup_{|x| \leq t}|u_n(x)|$ is bounded. Hence, by the Weierstrass M-test,
$\sum_{n\geq 0}u_n(x)$ converges uniformly on compact sets in $B_r(0)$. Since the uniform limit of harmonic functions on compact sets is harmonic by \cite[Theorem 1.23]{ABR}, (1) follows.

Conversely, assume that there is $u(x) \in H(B_r(0))$ such that $u(x)=\sum_{n\geq 0}u_n(x)$ as in Proposition \ref{prop_regular_harmonic}.
Let $s$ with $0<s<r$. Take $\rho>0$ with $r>\rho>s$.
Then, $u(\rho \xi)$ is a continuous function on $\xi \in S^{d-1}$, and thus,
\begin{align*}
\norm{u(\rho \bullet)}_{L^2(S^{d-1})}^2 = \sum_{n\geq 0} \rho^{2n} \norm{u_n}_{L^2(S^{d-1})}^2 \geq \rho^{2n} \norm{u_n}_{L^2(S^{d-1})}^2
\end{align*}
for any $n \geq 0$. Hence, 
\begin{align*}
q_s(u)= \sum_{n\geq 0} s^n \norm{u_n}_{L^2(S^{d-1})} \leq  \sum_{n\geq 0}(s/\rho)^n \norm{u(\rho \bullet)}_{L^2(S^{d-1})}
\leq \frac{1}{1-s/\rho}\norm{u(\rho \bullet)}_{L^2(S^{d-1})}<\infty.
\end{align*}
Moreover, $q_s(u) \leq \frac{1}{1-s/\rho}\norm{u(\rho \bullet)}_{L^2(S^{d-1})} \leq \frac{1}{1-s/\rho}\norm{u(\rho \bullet)}_{S^{d-1},\infty}
\leq \frac{1}{1-s/\rho}p_\rho(u)$, and \eqref{eq_seminorm_equiv} implies that there is $t<r$ and $C>0$ such that $p_s(u) \leq C q_t(u)$ for any $u \in H(B_r(0))$.
Hence, the equivalence of the topology follows.
\end{proof}
\begin{lem}
Let $r>0$ and let $(a_n)_{n\geq 0}$ be a sequence of real numbers.
Then the following are equivalent:
\begin{enumerate}
  \item[(1)] For any $s$ with $0 < s < r$, the series
$\sum_{n=0}^\infty s^n |a_n|$ converges.
  \item[(2)] For any $s$ with $0 < s < r$, $\sup_{n\geq 0} s^n |a_n| < \infty$.
\end{enumerate}
\end{lem}
\begin{proof}
Assume (1) holds. Fix $s$ with $0 < s < r$. By assumption, the series $\sum_{n=0}^\infty s^n |a_n|$
converges. Hence, $s^n |a_n| \to 0$ if $n\to\infty$.
In particular, the sequence $(s^n a_n)_{n\geq 0}$ is bounded. Hence, (2) holds.
Conversely, assume (2) holds.
Fix $s$ with $0 < s < r$. Choose $t$ such that $s < t < r$.
Applying (2), we obtain a constant $C>0$ such that $t^n |a_n| \leq C$ for any $n\geq 0$.
Hence, $\sum_{n=0}^\infty s^n |a_n| \leq \sum_{n=0}^\infty (s/t)^n t^n |a_n| \leq C \frac{1}{1-s/t} <\infty$.
\end{proof}

Hence, we can characterize $H(B_r(0))$ as
\begin{align}
H(B_r(0)) = \left\{(u)_{n\geq 0} \in \prod_{n \geq 0}\Harm_{d,n} \mid \sup_{n\geq 0}s^n \norm{u_n}_{L^2(S^{d-1})} < \infty \text{ for any }0<s<r \right\}.
\label{eq_app_harm_seq}
\end{align}

Next, we will study $H'(B_r(0))$. Let $T:H(B_r(0)) \rightarrow \R$ be a continuous linear map.
Since $\Harm_{d,n} \subset H(B_r(0))$, by restriction, we have a linear map $T|_{\Harm_{d,n}}:\Harm_{d,n} \rightarrow \R$. Hence, there are unique harmonic polynomials $t_n \in \Harm_{d,n}$ such that
\begin{align*}
(t_n,u_n)_{L^2(S^{d-1})} = T(u_n)\qquad\text{ for any }u_n \in \Harm_{d,n}.
\end{align*}
It is clear that $\pr_n:H(B_r(0)) \rightarrow \Harm_{d,n}$ is continuous, where $\Harm_{d,n}$ is a Hilbert space by $(-,-)_{L^2}$. Hence, $(t_n,\bullet)_{L^2(S^{d-1})}\circ \pr_n: H(B_r(0)) \rightarrow \R$ is a continuous map, which is simply denoted by $t_n$. Note that $t_n$ coincides with 
$\sum_{|\al|=n}a_\al \partial^\al \delta_0$ for some $a_\al \in \R$.

\begin{prop}\label{prop_equiv_dist}
For a sequence $(t_n)_n \in \prod_{n \geq 0} \Harm_{d,n}$, the following conditions are equivalent:
\begin{enumerate}
\item
There is a continuous linear map $T:H(B_r(0)) \rightarrow \R$ such that $(t_n,\bullet)_{L^2(S^{d-1})} = T \circ \pr_n$ for all $n \geq 0$.
\item
There is $\rho >0$ such that $r > \rho$ and 
\begin{align*}
\sup_{n \geq 0} \rho^{-n} \norm{t_n}_{L^2(S^{d-1})}<\infty.
\end{align*}
\end{enumerate}
\end{prop}
\begin{proof}
Assume that (2) holds. Then, there is $C>0$ such that $\norm{t_n}_{L^2}\leq C \rho^n$ for any $n\geq 0$.
Let $u =(u_n)_{n\geq 0} \in H(B_r(0))$. 
Then,
\begin{align*}
\sum_{n \geq 0} |(t_n,u_n)_{L^2(S^{d-1})}| \leq \sum_{n \geq 0} \norm{t_n}_{L^2(S^{d-1})}\norm{u_n}_{L^2(S^{d-1})} \leq  \sum_{n \geq 0} C \rho^n \norm{u_n}_{L^2(S^{d-1})}=C q_{\rho}(u) <\infty.
\end{align*}
Hence, the linear map $T:H(B_r(0)) \rightarrow \R$ defined by $T(u)=\sum_{n \geq 0} (t_n,u_n)_{L^2(S^{d-1})}$ is well-defined and continuous.

Conversely, let $T\in H'(B_r(0))$. Then, there is $C>0$ and $s$ with $r>s>0$ such that
\begin{align*}
|T(u)| \leq C q_s(u)\qquad \text{ for any }u\in H(B_r(0)).
\end{align*}
Hence, $|(t_n, u_n)_{L^2(S^{d-1})}|=|T(u_n)| \leq C q_s(u_n) = C s^n \norm{u_n}_{L^2(S^{d-1})}$.
By taking $u_n=t_n$, we have $\norm{t_n} \leq C s^n$ for any $n \geq 0$.
\end{proof}
\begin{cor}\label{cor_HD_explicit}
A sequence $t=(t_n)_{n \geq 0} \in \prod_{n \geq 0} \Harm_{d,n}$ defines a distribution in $H'(\bD)$ if and only if
there is $1>\rho>0$ and $C>0$ such that $\norm{t_n}_{L^2(S^{d-1})} < C \rho^n$.
In this case, $t \in H'(\bD)_0$ if and only if $t_0=0$.
\end{cor}

Let $d \geq 3$. The space $\Harm_d$ carries a natural irreducible highest weight representation of the conformal Lie algebra $\mathfrak{so}(d+1,1)$.
It admits, up to an overall scalar, a unique positive-definite
contravariant inner product. With the normalization $(1,1)_{\mathrm{CFT}}=1$,
this inner product is given by
\begin{align}
(u,v)_{\mathrm{CFT}}
=
\sum_{n\geq 0}
\frac{2n+d-2}{d-2}
(u_n,v_n)_{L^2(S^{d-1})}
\label{eq_inner_3}
\end{align}
for $u=\sum_{n\geq 0}u_n$ and $v=\sum_{n\geq 0}v_n$ in $\Harm_d$ (see \cite[Proposition 3.8, Proposition 3.2 and Lemma B.2]{MHigher}). We denote by $H_{\mathrm{CFT}}$ the Hilbert space completion of $\Harm_d$ with respect to this inner product
and by $\widehat{\Sym}(H_{\mathrm{CFT}})$ the corresponding real symmetric Fock space. As explained in Remark \ref{rem_AQFT}, after
complexification it is naturally identified with the Hilbert space of
states of the free massless scalar field in axiomatic quantum field theory.
\begin{rem}
\label{rem_AQFT}
We briefly recall the relation of $H_{\mathrm{CFT}}$ with the massless scalar field in axiomatic quantum field theory (for more detail, see \cite[Appendix C]{MHigher}). Let
\[
N_+
=
\left\{
q=(q_0,\ldots,q_{d-1})\in\mathbb{R}^{1,d-1}
\ \middle|\
q_0^2-\sum_{i=1}^{d-1}q_i^2=0,\quad q_0>0
\right\}
\]
be the forward light cone, equipped with the Lorentz-invariant measure $d\mu(q)=\frac{1}{2q_0}\,d^{d-1}q$. The complex Hilbert space $L^2(N_+,d\mu)$ inherits a 
unitary representation of the Poincar\'e group
$\mathcal{P}_d=\mathbb{R}^d\rtimes \mathrm{SO}_e(1,d-1)$ \cite{Wigner}
 by
$(\rho(a,\Lambda)f)(q)
=
e^{ia\cdot q}f(\Lambda^{-1}q)$.
This representation extends uniquely to a projective unitary representation of the conformal group $\mathrm{SO}_e(d,2)$ \cite{AngelopoulosLaoues} (see also \cite{SV,HL}).
Moreover, the $K$-finite vectors of this representation form an irreducible lowest weight representation of $\mathfrak{so}(d,2)_\C$
\cite[Theorem 9.7]{HSS}. After twisting by the involution on $\sod$ \cite[(3.3)]{MHigher}
and using the standard isomorphism
$\mathfrak{so}(d+1,1)_{\mathbb C}\cong \mathfrak{so}(d,2)_{\mathbb C}$,
this Harish-Chandra module is identified with $\Harm_d\otimes_{\mathbb R}\mathbb C$ (see
\cite[Proposition C.3]{MHigher}). Consequently,
$H_{\mathrm{CFT}}\otimes_{\mathbb R}\mathbb C
\cong
L^2(N_+,d\mu)$
as Hilbert spaces carrying the corresponding projective unitary representations.
Therefore, the symmetric Fock space
\[
\widehat{\Sym}\left(H_{\mathrm{CFT}}\right)\otimes_{\mathbb R}\mathbb C
\cong
\widehat{\Sym}\left(L^2(N_+,d\mu)\right)
\]
is the Hilbert space underlying the free massless scalar quantum field theory on Minkowski spacetime; see, for example, \cite[Chapter 10]{Arai}.
\end{rem}
\begin{thm}\label{thm_subspace}
Let $d \geq 3$.
Regard $H'(\bD)$ and $H_{\CFT}$ as subspaces of $\prod_{n\ge 0}\Harm_{d,n}$. 
Then,
\begin{align*}
H'(\bD) \subset H_\CFT
\quad\text{and}\quad
P(H'(\bD)) \subset \widehat{\mathrm{Sym}}(H_\CFT).
\end{align*}
\end{thm}
\begin{proof}
Assume that $t=(t_n)_{n \geq 0} \in \prod_{n \geq 0}\Harm_{d,n}$ defines a distribution in $H'(\bD)$.
Then, by Corollary \ref{cor_HD_explicit}, there are $1> \rho >0$ and $C>0$ such that $||t_n||_{L^2(S^{d-1})}<C \rho^n$.
Hence, 
\begin{align*}
\norm{t}_{\CFT}^2= \sum_{n \geq 0}\frac{2n+d-2}{d-2}\norm{t_n}_{L^2(S^{d-1})}^2 \leq C^2 \sum_{n \geq 0}\frac{2n+d-2}{d-2}\rho^{2n},
\end{align*}
which is finite since $\rho<1$. Hence, $t \in H_\CFT$, which implies $P(H'(\bD)) \subset \widehat{\mathrm{Sym}}(H_\CFT)$.
\end{proof}



It is noteworthy that the underlying Hilbert space of the massless free scalar field on Minkowski spacetime arises naturally from the prefactorization algebra on a Euclidean disk.
We also note that $H_{\CFT}$ is a reproducing kernel Hilbert space, and that the point evaluations
$\{\delta_a\}_{a\in\bD}\subset H'(\bD)\subset H_{\CFT}$ span a dense subspace of $H_{\CFT}$.
In the next section, we give an explicit description of the $\Disk$-algebra structure on this dense subspace.

In dimension $d=2$, the inner product \eqref{eq_inner_3} is no longer defined.
Instead, the complexification of $H'(\bD)_0$ admits a canonical embedding into a direct sum of Bergman spaces.
We investigate this two-dimensional counterpart in the next section as well.


\subsection{Algebra structure on harmonic distributions}\label{sec_algebra}
We denote by the same symbol \(T_U\) the natural isomorphism
\begin{align}
T_U:P(S^\cl(U))\longrightarrow P(H'(U))
\label{eq_T_U_sym}
\end{align}
induced by the linear isomorphism
$T_U:S^\cl(U)\longrightarrow H'(U)$ of Theorem~\ref{thm_compact_isomorphism}.
By $T_\bD$ and \eqref{def_pullback_algebra}, $P(H'(\bD))$ inherits a natural $\Disk$-algebra structure. We denote the products by $\m^H$. We now give an explicit description of this algebra structure.

Let $d \geq 2$.
Denote the isomorphism $T_\bD:P(\Cc(\bD)/\Delta \Cc(\bD)) \rightarrow P(H'(\bD))$ by $T$ for short.
For $f\in C^\infty_c(\bD)$, we write
\[
T_f=T([f])\in H'(\bD),
\]
where $[f]$ denotes the class of $f$ in
$C^\infty_c(\bD)/\Delta C^\infty_c(\bD)$.
Thus, for $u\in H(\bD)$, $T_f(u)$ denotes the evaluation of the
distribution $T_f$ on $u$.
Let $f \in \Cc(\bD)$ and $\phi \in \Disk(1)$. Then, for any $u \in H(\bD)$,
\begin{align}
\begin{split}
T_{W_\phi^\cl(f)}(u) &= \int_\bD \Om_{\phi}(\phi^{-1}(x))^{-\frac{d+2}{2}} f \circ \phi^{-1}(x)  u(x)d^dx\\
&= \int_{\bD} \Om_{\phi}(x')^{-\frac{d+2}{2}} f (x')  u(\phi(x')) \Om_\phi(x')^d d^dx'\\
&= \int_{\bD} f (x') \Om_{\phi}(x')^{\frac{d-2}{2}}u(\phi(x')) d^dx'.
\end{split}
\label{eq_maekara}
\end{align}
Therefore, we define a right action of $\phi$ on the space of harmonic functions $H(\bD)$ by
\begin{align*}
u(x)^\phi =\Om_{\phi}(x)^{\frac{d-2}{2}} u(\phi(x)).
\end{align*}
Here, by Proposition \ref{prop_Yamabe} and the assumption $u\in H(\bD)$, we have
$\Delta_{g_\std}\Om_{\phi}(x)^{\frac{d-2}{2}} u(\phi(x))= \Om_{\phi}(x)^{\frac{d+2}{2}} \Delta_{\phi^* g_\std} \phi^* u =0$,
and hence $u(x)^\phi \in H(\bD)$.
Define a left action of $\CE(1)$ on $H'(\bD)$ by
\begin{align}
(\phi\cdot T)(u)=T(u^\phi)
\label{eq_add_def_cdot1}
\end{align}
for $\phi\in\CE(1)$, $T\in H'(\bD)$ and $u\in H(\bD)$.
Then, we have
\begin{align}
T_{W_\phi^\cl(f)} =  \phi \cdot T_{f}.
\label{eq_add_def_cdot}
\end{align}
\begin{thm}\label{thm_disk_algebra}
Assume $d \geq 3$.
The $\Disk$-algebra $(P(H'(\bD)),\m^H)$ satisfies the following properties:
\begin{enumerate}
\item
For any $\phi \in \Disk(1)$, $\m_\phi^H: P(H'(\bD)) \rightarrow P(H'(\bD))$ is given by the left action of $\Disk(1)$ on $H'(\bD)$. In particular, it satisfies for any $a \in \bD$,
\begin{align*}
\m_\phi^H(\delta_{a}) = \Om_\phi(a)^{\frac{d-2}{2}} \delta_{\phi(a)}.
\end{align*}
\item
For any $(\phi,\psi) \in \Disk(2)$ and $a,b \in \bD$,
\begin{align*}
\m_{\phi,\psi}^H(\delta_a,\delta_b) = \Om_\phi(a)^{\frac{d-2}{2}}\Om_\psi(b)^{\frac{d-2}{2}}\delta_{\phi(a)}\cdot \delta_{\psi(b)}
 +\frac{\Om_\phi(a)^{\frac{d-2}{2}}\Om_\psi(b)^{\frac{d-2}{2}}}{(d-2)\om_d\norm{\phi(a)-\psi(b)}^{d-2}}.
\end{align*}
\item
Let $1>r,s>0$ and $a,b \in \bD$ satisfy
\begin{align*}
B_r(a) \cap B_s(b) = \emptyset\quad \text{ and }\quad
B_r(a), B_s(b) \subset B_1(0).
\end{align*}
Then, for any $\al,\be \in \N^d$, $(B_{a,r},B_{b,s}) \in \Disk(2)$ satisfies
\begin{align}
\begin{split}
\m_{B_{a,r},B_{b,s}}^H&(\pa^\al \delta_0,\partial^\be \delta_0)\\
&=
r^{|\al|+\frac{d-2}{2}}s^{|\be|+\frac{d-2}{2}} 
\left( (\partial^\al \delta_{a})\cdot (\partial^\be \delta_{b})+\pa_x^\al \pa_y^\be \frac{(-1)^{|\al|+|\be|}}{(d-2)\om_d \norm{x-y}^{d-2}}\Bigl|_{x-y=a-b}\right).
\end{split}
\label{eq_delta_harmonic_explicit}
\end{align}
\end{enumerate}
\end{thm}
\begin{proof}
(1) is clear. 
We choose representatives as in Lemma \ref{lem_dirac}.
Here we may assume, by taking $\ep>0$ sufficiently small, that 
$\overline{B_\ep(\phi(a))} \cap \overline{B_\ep(\psi(b))}=\emptyset.$
\begin{align*}
&\pa_{G_d}\left(T^{-1}(\rho_\phi^H(\delta_a)),T^{-1}(\rho_\psi^H(\delta_b))\right)\\
&=\Om_\phi(a)^{\frac{d-2}{2}}\Om_{\psi}(b)^{\frac{d-2}{2}}\pa_{G_d}\left(f_\ep^{\phi(a)}, f_\ep^{\psi(b)}\right)\\
&=\Om_\phi(a)^{\frac{d-2}{2}}\Om_{\psi}(b)^{\frac{d-2}{2}}
\int_{\R^d \times \R^d}\frac{1}{(d-2)\om_d \norm{x-y}^{d-2}}
f_\ep^{\phi(a)}(x) f_\ep^{\psi(b)}(y) d^dxd^dy\\
&=
\frac{\Om_\phi(a)^{\frac{d-2}{2}}\Om_{\psi}(b)^{\frac{d-2}{2}}}{(d-2)\om_d \norm{\phi(a)-\psi(b)}^{d-2}}.
\end{align*}
Hence, (2) follows from Theorem \ref{thm_classical_disk}.
(3) follows from
\begin{align*}
&(B_{a,r}\cdot \pa^\al \delta_0)(u(x)) = 
(\pa^\al \delta_0)\left(\Om_{B_{a,r}}(x)^{\frac{d-2}{2}}u(rx+a)\right)=(\pa^\al \delta_0)\left(r^{\frac{d-2}{2}}u(rx+a)\right)\\
&=(-1)^{|\al|} r^{\frac{d-2}{2}}\delta_0\left(\pa^\al u(rx+a)\right)=(-1)^{|\al|} r^{|\al|+\frac{d-2}{2}}(\pa^\al u)(a)
=r^{|\al|+\frac{d-2}{2}} (\partial^\al \delta_{a})(u(x))
\end{align*}
and
\begin{align*}
\pa_{G}(&r^{|\al|+\frac{d-2}{2}} \partial^\al \delta_{a},s^{|\be|+\frac{d-2}{2}} \partial^\be \delta_{b})\\
&=
r^{|\al|+\frac{d-2}{2}}s^{|\be|+\frac{d-2}{2}}
\int_{\R^d \times \R^d}\frac{1}{(d-2)\om_d \norm{x-y}^{d-2}}
(\partial^\al f_\ep^{a}(x))  (\partial^\be f_\ep^{b}(y)) d^dxd^dy\\
&=r^{|\al|+\frac{d-2}{2}}s^{|\be|+\frac{d-2}{2}}\int_{\overline{B_\ep(a)} \times \overline{B_\ep(b)}}(-1)^{|\al|+|\be|}\left(\pa_{y}^\be\pa_{x}^\al\frac{1}{(d-2)\om_d \norm{x-y}^{d-2}}\right)
f_\ep^{a}(x) f_\ep^{b}(y) d^d xd^d y\\
&=(-1)^{|\al|+|\be|}r^{|\al|+\frac{d-2}{2}}s^{|\be|+\frac{d-2}{2}}\left(\pa_{y}^\be\pa_{x}^\al\frac{1}{(d-2)\om_d \norm{x-y}^{d-2}}\right)\Bigl|_{x=a,y=b},\\
&=
(-1)^{|\al|+|\be|}r^{|\al|+\frac{d-2}{2}}s^{|\be|+\frac{d-2}{2}} \pa_x^\al \pa_y^\be \frac{1}{(d-2)\om_d \norm{x-y}^{d-2}}\Bigl|_{x-y=a-b}.
\qedhere
\end{align*}
\end{proof}

For the remainder of Section~3, in the case $d=2$,
all vector spaces are taken over $\mathbb C$.
In Theorem \ref{thm_disk_algebra}, there appears the factor $\frac{1}{(d-2)\om_d}$ in the normalization of the Green function \eqref{eq_def_Green}, but this can be removed by rescaling the product (physically, a field renormalization). Namely, consider the following linear isomorphism $\cN:P(H'(\bD)) \rightarrow P(H'(\bD))$
\begin{align}
f_1 \dots f_n \mapsto 
\begin{cases}
(4\pi)^{n/2}i^n f_1\dots f_n & d =2\\
\left((d-2)\om_d\right)^{n/2} f_1\dots f_n & d \geq 3
\end{cases}
\label{eq_normalization}
\end{align}
and, via this isomorphism, we can consider the normalized $\CE$-algebra structure on $P(H'(\bD))$, which is isomorphic to the original $\CE$-algebra by $\cN$. We write this product as $\m^{NH}$.
The normalized product and the original product differ by the factor in \eqref{eq_normalization} only when the operation changes the degree.
For example, \eqref{eq_delta_harmonic_explicit} becomes as follows:
\begin{align}
\m_{T_{x_0}r^D, T_{y_0}s^D}^{NH}&(\pa^\al \delta_0,\partial^\be \delta_0)
&=
r^{|\al|+\frac{d-2}{2}}s^{|\be|+\frac{d-2}{2}} 
\left( (\partial^\al \delta_{x_0})\cdot (\partial^\be \delta_{y_0})+\pa_x^\al \pa_y^\be \frac{(-1)^{|\al|+|\be|}}{\norm{x-y}^{d-2}}\Bigl|_{x-y=x_0-y_0}\right).
\label{eq_delta_harmonic_explicit_normal}
\end{align}

%
%

Next, we consider the case $d=2$, where we need the correction by the harmonic cocycles (Theorem \ref{thm_classical_disk2}).
Note that the space of harmonic polynomials over $\C$ in $d=2$ is
\begin{align*}
\Harm_{2,n}\otimes_\R \C = \C z^n \oplus \C \z^n,\quad \text{for $n>0$}.
\end{align*}
Note that 
\begin{align*}
(z^n,z^m)_{L^2(S^1)} = \int_{S^1}\overline{z^n} z^m \frac{d\theta}{2\pi} = \delta_{n,m}
\end{align*}
and similarly $(z^n,\z^m)_{L^2(S^1)}=0$ for any $m>0$. Hence, the distribution corresponding to $z^n$ (resp. $\z^n$), denote by $T_{z^n}$ (resp. $T_{\z^n}$), satisfies
\begin{align*}
T_{z^n} =\frac{(-1)^n}{n!} \partial_z^n \delta_0,\qquad T_{\z^n} = \frac{(-1)^n}{n!}\partial_\z^n \delta_0.
\end{align*}
for any $n \geq 0$ as elements in $H'(\bD)$ (see Proposition \ref{prop_equiv_dist}).
Let $\phi \in \CEt(1)$ and $a\in \bD$.
Then,
\begin{align}
(\phi \cdot \pa_z^n \delta_{a})(u) &= \int_\bD (\pa_z^n \delta_{a}) u(\phi(z)) d^2z=(-1)^n \int_\bD \delta_{a} (\pa_z^n u(\phi(z))) d^2z
\label{eq_add_transform}
\end{align}
In particular, for $n=1$, this is equal to
$- \int_{\bD} \delta_{a} \phi'(z) u'(\phi(z)) dz^2 = - \phi'(a) u'(\phi(a))$.
Hence, we have
\begin{align*}
\phi\cdot \pa_z \delta_a = \phi'(a) \pa_z \delta_{\phi(a)}.
\end{align*}
Such an action describes the monoid action without the cocycle correction.

By Lemma \ref{lem_dirac}, for any $a\in \bD$, $\pa_z^{n+1}\delta_a$ is represented by $\pa_z^{n+1}f_\ep^a \in \Cc(\bD)$. Hence,
\begin{align*}
\pa_{H_\phi} (T^{-1}(\pa_z^{n+1}\delta_a),T^{-1}(\pa_z^{m+1}\delta_b))&=
\pa_{H_\phi} (\pa_z^{n+1}f_\ep^a,\pa_z^{m+1}f_\ep^b)\\
&= \int_{\bD \times \bD} H_\phi(z,w) \pa_z^{n+1}f_\ep^a(z) \pa_w^{m+1}f_\ep^b(w) d^2zd^2w\\
&=(-1)^{n+m} \int_{\bD \times \bD} \pa_z^{n+1} \pa_w^{m+1}H_\phi(z,w) f_\ep^a(z) f_\ep^b(w) d^2zd^2w\\
&=(-1)^{n+m} \pa_z^{n+1} \pa_w^{m+1}H_\phi(z,w)|_{z=a,w=b}.
\end{align*}
Here, since $\pa_z \pa_w H_\phi(z,w)$ is a holomorphic function on $\bD^2$, we can expand it into a power series as follows:
\begin{align*}
\frac{1}{2\pi}
\pa_z \pa_w H_\phi(z,w) &= \frac{1}{4\pi}\left(\frac{\phi'(z)\phi'(w)}{(\phi(z)-\phi(w))^2} - \frac{1}{(z-w)^2}\right)\\
&=\frac{1}{4\pi} \sum_{n,m \geq 0}A_{n,m}(\phi)z^n w^m.
\end{align*}
From this, it follows that $\frac{1}{2\pi}\pa_{H_\phi} (T^{-1}(\pa_z^{n+1}\delta_0),T^{-1}(\pa_z^{m+1}\delta_0))=\frac{(-1)^{n+m}n!m!}{4\pi} A_{n,m}(\phi)$.


For any $n,m \geq 0$ and $(\phi,\psi) \in \CEt(2)$,
\begin{align*}
\pa_{G_2}&\left(
T^{-1}(\phi\cdot\partial_z^{n+1}\delta_a),
T^{-1}(\psi\cdot\partial_z^{m+1}\delta_b)
\right)\\
&= \frac{1}{-2\pi}\int_{\bD^2} \log|z-w| \pa_z^{n+1} f_\ep^a(\phi^{-1}(z))\pa_w^{m+1} f_\ep^b(\psi^{-1}(w)) |\phi'(\phi^{-1}(z))|^{-2} |\psi'(\psi^{-1}(w))|^{-2}dz^2dw^2\\
&= \frac{1}{-2\pi}\int_{\bD^2} \log|\phi(z)-\psi(w)| \pa_z^{n+1} f_\ep^a(z)\pa_w^{m+1} f_\ep^b(w) dz^2dw^2\\
&= \frac{(-1)^{n+m}}{-2\pi}\pa_z^{n+1} \pa_w^{m+1} \log|\phi(z)-\psi(w)| \Bigl|_{z=a,w=b}\\
&= \frac{(-1)^{n+m}}{-4\pi}\pa_z^n \pa_w^m\frac{\phi'(z)\psi'(w)}{(\phi(z)-\psi(w))^2}\Bigl|_{z=a,w=b}.
\end{align*}
Note that the factor $-4\pi$ is removed by the normalization \eqref{eq_normalization}.
Similarly, since $\pa_z \pa_{\bar{w}} H_\phi(z,w)=0$ and $\pa_z \pa_{\bar{w}} \frac{\phi'(z)\psi'(w)}{(\phi(z)-\psi(w))^2}=0$,
it follows that $\pa_{H_\phi} (T^{-1}(\pa_z^{n+1}\delta_a),T^{-1}(\pa_\z^{m+1}\delta_b))=0$
and $\pa_{G_2}(
\phi\cdot\partial_z^{n+1}\delta_a,
\psi\cdot\partial_{\bar z}^{m+1}\delta_b
)=0$.
Hence, we have:
\begin{thm}\label{thm_disk_algebra2}
The normalized $\Diskt$-algebra $(P(H'(\bD)_0),\m^{NH})$ satisfies the following properties:
\begin{enumerate}
\item
$H'(\bD)_0$ consists of all the sums $\sum_{n \geq 0}a_n \frac{1}{(n+1)!}\pa_z^{n+1}\delta_0 +b_n \frac{1}{(n+1)!}\pa_\z^{n+1}\delta_0 \in \prod_{n>0}\Harm_{2,n}$ that satisfy the growth condition in Corollary \ref{cor_HD_explicit}.
\item
For any $\phi \in \Diskt(1)$ and $a,b \in \bD$,
\begin{align*}
W_\phi^\cl (\pa_z \delta_a) &= \phi'(a) \pa_z \delta_{\phi(a)}\\
\pa_\phi(\pa_z \delta_a, \pa_z \delta_b)&=\frac{\phi'(a)\phi'(b)}{(\phi(a)-\phi(b))^2} - \frac{1}{(a-b)^2}\\
\pa_\phi(\pa_\z \delta_a, \pa_\z \delta_b)&=\overline{\frac{\phi'(a)\phi'(b)}{(\phi(a)-\phi(b))^2}} - {\frac{1}{(\bar{a}-\bar{b})^2}}\\
\pa_{\phi}(\pa_z^{n+1} \delta_a, \pa_\z^{m+1}\delta_b)&=0
\end{align*}
for any $n,m \geq 0$, and $\m_\phi^{NH}:P(H'(\bD)_0) \rightarrow P(H'(\bD)_0)$ is given by $W_\phi^\cl \circ \exp(\pa_\phi)$.
\item
For any $(\phi,\psi) \in \Diskt(2)$,
\begin{align*}
\m_{\phi,\psi}^{NH}(\pa_z \delta_a, \pa_z \delta_b)=
\phi'(a)\psi'(b)\pa_z \delta_{\phi(a)} \cdot \pa_z \delta_{\psi(b)} +
\frac{\phi'(a)\psi'(b)}{(\phi(a)-\psi(b))^2}.
\end{align*}
\end{enumerate}
\end{thm}

We note that $\pa_\phi(\pa_z\delta_a,\pa_z\delta_b)
=2\pa_z\pa_wH_{\phi}(z,w)|_{z=a,w=b}$
is the Grunsky kernel in \eqref{eq_def_Grunsky}, and by Remark \ref{rem_Grunsky}, its restriction to the diagonal is $\frac{1}{6}S(\phi)(a)$.

For $a\in\bD$, set
\begin{align}
T(a)
=
\frac12\pa_z\delta_a\cdot\pa_z\delta_a.
\label{eq_add_def_tensor}
\end{align}
Applying Theorem~\ref{thm_disk_algebra2}, we obtain
\begin{align}
\begin{split}
\m_\phi^{NH}(T(a))
&=
W_\phi^\cl
\left(
\frac12\pa_z\delta_a\cdot\pa_z\delta_a
\right)
+
\pa_\phi
\left(
\frac12\pa_z\delta_a\cdot\pa_z\delta_a
\right)\\
&=
\frac12\phi'(a)^2
\left(
\pa_z\delta_{\phi(a)}\cdot
\pa_z\delta_{\phi(a)}
\right)
+
\pa_z\pa_wH_{\phi}(z,w)|_{z=w=a}\\
&=
\phi'(a)^2T(\phi(a))
+
\frac1{12}S(\phi)(a).
\end{split}
\label{eq_T_transformation}
\end{align}
Thus, \eqref{eq_T_transformation} is precisely the standard transformation law
\eqref{eq_standard_T_transformation} of a holomorphic energy-momentum tensor
with central charge $c=1$.
We return to this point in Section~\ref{sec_Ward}, where we
complete this identification by proving the conformal Ward identity.

We conclude this section by explaining the embedding of $H'(\bD)_0$ into a Hilbert space.
Recall that the Bergman space $A^2(\bD)$ is the Hilbert space of square-integrable holomorphic functions on $\bD$, with orthonormal basis $\left\{{\sqrt{n+1}}z^n\right\}_{n\geq 0}$ (see \cite[Section 1.2]{MBergman}),
and let $\overline{A^2(\bD)}$ denote its anti-holomorphic counterpart.
With respect to the expansion in Theorem \ref{thm_disk_algebra2} (1), define a linear map $H'(\bD)_0 \rightarrow  A^2(\bD)\oplus \overline{A^2(\bD)}$ by
\begin{align*}
\sum_{n \geq 0}a_n \frac{1}{(n+1)!}\pa_z^{n+1}\delta_0 +b_n \frac{1}{(n+1)!}\pa_\z^{n+1}\delta_0 &\mapsto 
\left(\sum_{n \geq 0}a_n z^n,\sum_{n \geq 0}b_n\z^{n}\right)
\end{align*}
This map is injective and has dense image. The naturality of this embedding is clarified in \cite{MBergman}.

\begin{rem}\label{rem_log}
The vector $\delta_0\in H'(\bD)$, which is excluded in this construction, corresponds to the logarithmic field in the two-dimensional massless free scalar conformal field theory (often denoted $X(z,\bar z)$ in string theory) \cite{Polc}.
This theory is non-unitary, and $X(z,\bar z)$ does not admit a global conformal symmetry (i.e.\ it is not quasi-primary); from the perspective of this paper, this phenomenon appears in Proposition~\ref{prop_Weyl_anomaly}.
In the language of vertex operator algebras, logarithmic fields can be described using a notion of logarithmic full vertex algebras; see \cite{ACM}.
\end{rem}

%
%

In our papers \cite{MBergman} and \cite{MLeft}, we prove that, in the cases $d =2$ and $d \geq 3$ respectively, the $\CE$-algebra structure
can be lifted to its ind-Hilbert space completion if we refine the definition of the operad $\CE$. Whether such a completion can be obtained geometrically and homologically within the framework of factorization algebras is an important question.

\subsection{Conformal Ward identity}
\label{sec_Ward}
In Section~\ref{sec_algebra}, we introduced the field
\[
T(w)=\frac12\pa_z\delta_w\cdot\pa_z\delta_w,
\qquad w\in\bD,
\]
and showed that its conformal transformation law agrees with that of a
holomorphic energy-momentum tensor of central charge $c=1$ (see
\eqref{eq_T_transformation} and \eqref{eq_standard_T_transformation}).
To complete this identification, we now establish the conformal Ward identity,
showing that $T$ generates infinitesimal conformal transformations.

Let $X=v(z)\pa_z$ be a holomorphic vector field on $\bD$.
Locally, $X$ generates a flow $\Psi_t^X$, and for sufficiently small $t$
the maps $\Psi_t^X$ are injective holomorphic maps.
By the functoriality of the prefactorization algebra, these maps induce transformations of observables, whose derivative at $t=0$
gives the infinitesimal conformal action.
The conformal Ward identity asserts that this infinitesimal action is
generated by residues of the factorization product with the energy-momentum tensor
(see also the standard formulation in the physics literature \cite[Section 5.2.2]{FMS}).
To formulate the Ward identity rigorously, we work on the subspace
$\Sym(\Hpt)\subset P(H'(\bD)_0)$ introduced below.
On this subspace, the derivatives and residues appearing in the Ward identity
are shown to exist and are characterized weakly by pairing with holomorphic
functions.
The anti-holomorphic sector is completely analogous, so we restrict below to the holomorphic sector.

For $m\geq 0$ and $a\in\bD$, set
\begin{align*}
J^{(m)}(a)
=
(-1)^{m+1}\pa_z^{m+1}\delta_a,\quad J(a)=J^{(0)}(a)
\end{align*}
and
\begin{align*}
\Hpt= \mathrm{Span}_{\C}
\left\{
J^{(m)}(a)
\ \middle|\
a\in\bD,\ m\geq0
\right\}
\subset H'(\bD)_0.
\end{align*}

Let $\mathcal O(\bD)$ denote the space of holomorphic functions on $\bD$.
Since holomorphic functions are harmonic, evaluation of distributions
defines a pairing
\begin{align}
\langle-,-\rangle:
\mathcal O(\bD)\times\Hpt
\longrightarrow\C,
\quad
\langle f,A\rangle = A(f),
\label{eq_app_pairing}
\end{align}
which satisfies
\begin{align}
\langle f,J^{(m)}(a)\rangle
=
\pa_z^{m+1}f(a).
\label{eq_ev_Jf}
\end{align}
\begin{lem}\label{lem_app_separating}
If $A\in\Hpt$ satisfies
$\langle f,A\rangle=0$
for any $f \in\mathcal O(\bD)$, then $A=0$.
Moreover, $\{J^{(m)}(a)\}_{a\in\bD,m\geq 0}$ forms a basis of $\Hpt$.
\end{lem}
\begin{proof}
By Hermite interpolation, given finitely many points and finitely many values of derivatives at those points, there exists a polynomial realizing the prescribed data. Hence, the assertion follows from \eqref{eq_ev_Jf}.
\end{proof}

By symmetrization, the pairing \eqref{eq_app_pairing} extends degreewise to
$\langle-,-\rangle:
\Sym\mathcal O(\bD)\times\Sym\Hpt
\rightarrow\C$.
By Lemma~\ref{lem_app_separating}, this pairing separates
$\Sym\Hpt$.

\begin{dfn}\label{def_weak_smooth}
\begin{enumerate}
\item
A family $W_t\in\Sym\Hpt$ depending on a real parameter $t$ is called
{\it smooth} if $\langle u,W_t\rangle$ is smooth in $t$ for any $u\in\Sym\OD$.
\item
Let $W_t\in\Sym\Hpt$ be a smooth family.
We say that its derivative at $t=0$ exists in $\Sym\Hpt$ if there is
$A\in\Sym\Hpt$ such that
\[
\left.\frac{d}{dt}\right|_{t=0}\langle u,W_t\rangle
=
\langle u,A\rangle\qquad\text{ for any $u\in\Sym\OD$}.
\]
\item
Let $S\subset\bD$ be a finite subset and set $\bD_S=\bD\setminus S$.
A family $W_z\in\Sym\Hpt$ $(z\in\bD_S)$ is called
{\it meromorphic on $\bD$ with poles only in $S$} if
$\langle u,W_z\rangle$ is meromorphic on $\bD$ with poles contained in $S$
for any $u\in\Sym\OD$.
For $a\in S$, we say that $\Res_{z=a}W_z\,dz$ exists in $\Sym\Hpt$
if there is $B\in\Sym\Hpt$ such that
\[
\Res_{z=a}\langle u,W_z\rangle\,dz
=
\langle u,B\rangle
\qquad\text{ for any $u\in\Sym\OD$}.
\]
\end{enumerate}
\end{dfn}
By Lemma~\ref{lem_app_separating}, the elements $A$ and $B$ in (2) and (3) are unique if they exist. We denote them by
$\left.\frac{d}{dt}\right|_{t=0}W_t$ and $
\Res_{z=a}W_z\,dz$, respectively.
Set
\[
\Vh
=
\{v(z)\pa_z\mid v\in\OD\},
\]
the Lie algebra of holomorphic vector fields on $\bD$. 

Let $X=v(z)\pa_z\in\Vh$.
Recall that for a relatively compact open subset $U\subset \bD$, by the existence theorem for flows, there exist $\ep>0$ and a unique
smooth map
\[
\Psi^X:(-\ep,\ep)\times U\longrightarrow\bD,
\qquad
(t,z)\longmapsto\Psi_t^X(z),
\]
such that
\[
\frac{d}{dt}\Psi_t^X(z)
=
v(\Psi_t^X(z)),
\qquad
\Psi_0^X(z)=z.
\]
For sufficiently small $\ep$, each
$\Psi_t^X:U\rightarrow \bD$ is injective, and since $X$ is holomorphic,
$\Psi_t^X$ is holomorphic for each fixed $t$
(see, for example, \cite[Section~1.9]{Forstneric}).

We first study the infinitesimal action of the flow $\Psi_t^X$ on $\Sym \Hpt$ in the sense of Definition~\ref{def_weak_smooth}.
For $m_i\geq0$ and $a_i\in\bD$ $(1\leq i\leq n)$, set
\[
S=\{a_1,\dots,a_n\},
\qquad
\bmJ
=
J^{(m_1)}(a_1)\cdots J^{(m_n)}(a_n)
\in\Sym\Hpt.
\]
Choose a relatively compact open neighborhood  $U_S \subset \bD$ of $S$.
The point-supported distribution
$J^{(m_i)}(a_i) = (-1)^{m_i+1}\pa_z^{m_i+1}\delta_{a_i}$ can be regarded naturally as an element of $H'(U_S)$.
Let $\bmJ_{U_S}$ denote $\bmJ$ regarded as an element of
$P(H'(U_S))$.
For $t\in(-\ep,\ep)$, define a linear map
$(\Psi_{t}^X)_*:P(H'(U_S))\longrightarrow P(H'(\bD))$ by
\begin{align}
P(H'(U_S))
&\xrightarrow{\,T_{U_S}^{-1}\circ \mathcal N\,}
P(S^\cl(U_S))
\xrightarrow{\,S^\cl(\Psi_{t}^X)\,}
P(S^\cl(\bD))
\xrightarrow{\,\mathcal N^{-1}\circ T_{\bD}\,}
P(H'(\bD)),
\label{eq_app_def_global}
\end{align}
(see Remark \ref{rem_extend_disk2}). We set
\[
(\Psi_{t}^X)_*\bmJ
=
(\Psi_{t}^X)_*(\bmJ_{U_S}).
\]
The independence of $U_S$ is proved later in Proposition~\ref{prop_app_global_explicit}.
By Remark~\ref{rem_extend_disk2},
\begin{align}
S^\cl(\Psi_{t}^X)
=
W_{\Psi_{t}^X}^\cl
\circ
\exp\left(
-(2\pi)^{-1}\partial_{\tilde{H}_{\Psi_{t}^X}}
\right).
\label{eq_add_Psi_Scl}
\end{align}
We first examine the contribution of
$W_{\Psi_{t}^X}^\cl$ to the infinitesimal action.
\begin{lem}\label{lem_app_global_one_particle}
For $a\in U_S$ and $m\geq0$, 
$(\Psi_{t}^X)_*J^{(m)}(a)\in\Hpt$
defines a smooth family and
\begin{align}
\left.\frac{d}{dt}\right|_{t=0}
(\Psi_{t}^X)_*J^{(m)}(a)
=
\sum_{k=0}^{m+1}
\binom{m+1}{k}
v^{(k)}(a)J^{(m+1-k)}(a).
\label{eq_app_cl_inf}
\end{align}
\end{lem}

\begin{proof}
The contraction operators $\pa_{\tilde{H}_\phi}$ in \eqref{eq_add_Psi_Scl} act trivially
on elements of degree one.  Hence, as in
\eqref{eq_maekara}, for every $u\in H(\bD)$,
\begin{align}
\left\langle
u,(\Psi_{t}^X)_*J^{(m)}(a)
\right\rangle
=
\left\langle
u\circ \Psi_{t}^X, J^{(m)}(a)
\right\rangle
=
\left.
\pa_z^{m+1}(u\circ\Psi_{t}^X)(z)
\right|_{z=a}.
\label{eq_app_derivative_Hpt}
\end{align}
Since $\Psi_{t}^X$ is holomorphic, the right-hand side is a finite linear combination of
$\{(\pa_z^k u)(\Psi_{t}^X(a))\}_{1\leq k\leq m+1}$ with coefficients depending polynomially on
$(\pa_z^j\Psi_{t}^X)(a)$.
Thus $(\Psi_{t}^X)_*J^{(m)}(a)$ is a finite linear combination of
$J^{(k)}(\Psi_{t}^X(a))$, and hence belongs to $\Hpt$
and depends smoothly on $t$.
To compute its derivative at $t=0$, it is enough to pair with
$f\in\OD$.
Then, $\left.\frac{d}{dt}\right|_{t=0}
f(\Psi_{t}^X(z))
=v(z)\pa_z f(z)$. Hence, we obtain
\begin{align}
\left.
\frac{d}{dt}
\right|_{t=0}
\pa_z^{m+1}(f\circ\Psi_{t}^X)(a)
&=
\pa_z^{m+1}\bigl(v\,\pa_z f\bigr)(a)=
\sum_{k=0}^{m+1}
\binom{m+1}{k}
v^{(k)}(a)\,
\pa_z^{m+2-k}f(a).
\label{eq_Lvcl_first}
\end{align}
Hence, \eqref{eq_app_cl_inf} follows from \eqref{eq_ev_Jf}.
\end{proof}

Motivated by Lemma~\ref{lem_app_global_one_particle}, define
$L_X^\cl:\Hpt\rightarrow\Hpt$ by
\[
L_X^\cl\bigl(J^{(m)}(a)\bigr)
=
\sum_{k=0}^{m+1}
\binom{m+1}{k}
v^{(k)}(a)J^{(m+1-k)}(a),
\]
and extend $L_X^\cl$ to $\Sym\Hpt$ as a derivation. 

We next consider the infinitesimal contribution of
$\partial_{\tilde{H}_{\Psi_{t}^X}}$. For $t\in(-\ep,\ep)$, set
\begin{align}
\begin{split}
K_t(z,w)
&=2 \pa_z \pa_w \tilde{H}_{\Psi_{t}^X}(z,w)=
\frac{(\Psi_{t}^X)'(z)(\Psi_{t}^X)'(w)}
     {(\Psi_{t}^X(z)-\Psi_{t}^X(w))^2}
-\frac{1}{(z-w)^2},\\
B_X(z,w)
&=
\left.\frac{d}{dt}\right|_{t=0}K_t(z,w)=
\frac{v'(z)+v'(w)}{(z-w)^2}
-2\frac{v(z)-v(w)}{(z-w)^3},
\end{split}
\label{eq_app_def_Bv}
\end{align}
By $\pa_z\pa_w\tilde{H}_\phi
 =\pa_z\pa_w H_\phi$, $K_t$ is the Grunsky kernel of $\Psi_{t}^X$ introduced in Remark~\ref{rem_Grunsky}. Thus $B_X$ is its infinitesimal variation along the flow.
By \eqref{eq_diagonal_Schwarzian},
\begin{align}
\begin{split}
K_t(z,z)&= \frac{1}{6}S\left(\Psi_{t}^X\right)(z),\quad B_X(z,z) =\frac16v'''(z).
\end{split}
\label{eq_app_KB}
\end{align}

Let $\Pair(n)$ denote the set of all collections $P$ of pairwise
disjoint two-element subsets of $\{1,\dots,n\}$. For
$P\in\Pair(n)$, set $|P|=\bigcup_{\{i,j\}\in P}\{i,j\}$, including $P=\emptyset$.
\begin{prop}\label{prop_app_global_explicit}
For every $t\in(-\ep,\ep)$,
\begin{align}
(\Psi_{t}^X)_*\bmJ
=
\sum_{P\in\Pair(n)}
\left(
\prod_{\{i,j\}\in P} \pa_{a_i}^{m_i}\pa_{a_j}^{m_j}K_t(a_i,a_j)
\right)
\left(
\prod_{k\notin |P|}J_t^{(m_k)}(a_k)
\right),
\label{eq_app_global_action}
\end{align}
where $J_{t}^{(m_i)}(a_i)=(\Psi_{t}^X)_*J^{(m_i)}(a_i)$.
In particular, $(\Psi_{t}^X)_*\bmJ$ belongs to $\Sym\Hpt$,
is independent of the choice of $U_S$, and depends smoothly on $t$.
Moreover,
\begin{align}
\left.\frac{d}{dt}\right|_{t=0}
(\Psi_{t}^X)_*\bmJ
=
L_X^\cl(\bmJ)
+
\sum_{1\leq i<j\leq n}
\pa_{a_i}^{m_i}\pa_{a_j}^{m_j}
B_X(a_i,a_j)\,
\bmJ_{\widehat{i,j}}.
\label{eq_app_diff_formula}
\end{align}
Here, we set  $\bmJ_{\widehat I}= \prod_{k\notin I}J^{(m_k)}(a_k)$ for $I\subset\{1,\dots,n\}$.
\end{prop}

\begin{proof}
By \eqref{eq_add_Psi_Scl} and Lemma \ref{lem_app_global_one_particle}, it remains to consider the contraction $-(2\pi)^{-1}\pa_{\tilde{H}_{\Psi_{t}^X}}$.
Note that, under the normalization map $\mathcal N$ in \eqref{eq_normalization},
this acts as $2\pa_{\tilde{H}_{\Psi_{t}^X}}$.
By \eqref{eq_app_def_Bv},
\begin{align*}
2\pa_{\tilde{H}_{\Psi_{t}^X}}(T_{U_S}^{-1}(J^{m_i}(a_i)),T_{U_S}^{-1}(J^{m_j}(a_j)))
&=2 (-1)^{m_i+m_j} \pa_{\tilde{H}_{\Psi_{t}^X}}\left(\pa_z^{m_i+1}f_\ep^{a_i}(z),
\pa_z^{m_j+1}f_\ep^{a_j}(z)\right)\\
&=\pa_{a_i}^{m_i}\pa_{a_j}^{m_j}K_t(a_i,a_j).
\end{align*}
Hence, expanding the exponential gives \eqref{eq_app_global_action}.
Since $K_t(z,w)$ depends smoothly on
$t$, by Lemma~\ref{lem_app_global_one_particle}, the right-hand side of \eqref{eq_app_global_action} belongs to
$\Sym\Hpt$ and depends smoothly on $t$.
Moreover, by uniqueness of the local flow, 
\eqref{eq_app_global_action} is independent of the
choice of $U_S$.
Differentiating \eqref{eq_app_global_action} at $t=0$, we obtain \eqref{eq_app_diff_formula}.
In fact, since $\Psi_0^X=\id$, we have $K_0(z,w)=0$ and thus $\pa_{z}^{n}\pa_{w}^{m}K_0(z,w)=0$.
Hence, pairings containing two or more pairs do not contribute to the first derivative at $t=0$.
\end{proof}

%

We next express the infinitesimal action \eqref{eq_app_diff_formula} in terms of residues of the factorization product with $T(w)$.
Let $w\in \bD_S=\bD\setminus S$.
After shrinking $U_S$ if necessary, choose an open neighborhood
$V\subset\bD$ of $w$ such that $V\cap U_S=\emptyset$.
Regarding $T(w)$ and $\bmJ$ as elements of $P(H'(V))$ and
$P(H'(U_S))$, respectively, define
\[
T(w)*\bmJ\in P(H'(\bD))
\]
to be their image under
\begin{align*}
P(H'(V))\otimes P(H'(U_S))
&\xrightarrow{
(e^{-\pa_{G_2}}\otimes e^{-\pa_{G_2}})
(T_V^{-1}\otimes T_{U_S}^{-1})
(\mathcal N\otimes\mathcal N)
}
\HY(V)\otimes\HY(U_S)
\\
&\xrightarrow{\HY(\iota)}
\HY(\bD)
\xrightarrow{\,
\mathcal N^{-1}\circ T_{\bD}\circ e^{\pa_{G_2}}
\,}
P(H'(\bD)),
\end{align*}
where $\iota:V\sqcup U_S\hookrightarrow\bD$ is the inclusion.
Since $\iota$ is an inclusion, $\tilde{H}_\iota=0$ by
Remark~\ref{rem_extend_disk2}.
Thus, taking into account the normalization map $\mathcal N$ in
\eqref{eq_normalization}, to compute $T(w)*\bmJ$ it suffices to consider the contraction by $-4\pi\pa_{G_2}$ appearing in \eqref{eq_why_cont}.
Since
$-4\pi\pa_{G_2}
\left(
T_V^{-1}(J(w)),
T_{U_S}^{-1}(J^{(m_i)}(a_i))
\right)
=
\pa_{a_i}^{m_i}\frac{1}{(w-a_i)^2}$, 
the product gives
\begin{align}
T(w)*\bmJ
={}&
T(w)\bmJ
+
\sum_{i=1}^n
\left(
\pa_{a_i}^{m_i}\frac1{(w-a_i)^2}
\right)
J(w)\bmJ_{\widehat i}
\nonumber\\
&+
\sum_{1\leq i<j\leq n}
\left(
\pa_{a_i}^{m_i}\frac1{(w-a_i)^2}
\right)
\left(
\pa_{a_j}^{m_j}\frac1{(w-a_j)^2}
\right)
\bmJ_{\widehat{i,j}}.
\label{eq_app_Tproduct}
\end{align}
In particular, the right-hand side is independent of the choices of
$V$ and $U_S$, and $T(w)*\bmJ$ is a $\Sym\Hpt$-valued meromorphic family
on $\bD$, with poles contained in $S$.

\begin{thm}\label{thm_Ward_identity}
For $X=v(z)\pa_z \in\Vh$, the residues below exist and satisfy
\begin{align}
\left.\frac{d}{dt}\right|_{t=0}
(\Psi_{t}^X)_*\bmJ
=\sum_{a\in S}
\Res_{w=a}
v(w)\bigl(T(w)*\bmJ\bigr)\,dw.
\label{eq_app_Ward}
\end{align}
\end{thm}
\begin{proof}
We compute the residues of the three terms in
\eqref{eq_app_Tproduct} separately.
The first term $T(w)\bmJ$ is holomorphic in $w$, and hence does not
contribute to the residue.
For the terms with one contraction, we claim that
\begin{align}
\Res_{w=a_i}v(w)\left(\pa_{a_i}^{m_i}\frac1{(w-a_i)^2}\right)J(w)\,dw
=
L_X^\cl\bigl(J^{(m_i)}(a_i)\bigr).
\label{eq_Ward_single}
\end{align}
Indeed, for $f\in\OD$, using
$\langle f,J(w)\rangle=f'(w)$, we obtain
\begin{align*}
\Res_{w=a_i}\left\langle
f,v(w)\left(\pa_{a_i}^{m_i}\frac1{(w-a_i)^2}\right)J(w)\,dw
\right\rangle
&=
\Res_{w=a_i}
v(w)f'(w)\left(\pa_{a_i}^{m_i}\frac1{(w-a_i)^2}\right)dw\\
&=
\Res_{w=a_i}
\pa_w^{m_i+1}\left(v(w)f'(w)\right)\frac1{w-a_i}dw\\
&=
\left.
\pa_w^{m_i+1}\bigl(v(w)f'(w)\bigr)
\right|_{w=a_i}.
\end{align*}
Hence \eqref{eq_Ward_single} follows from  \eqref{eq_Lvcl_first}.
For the terms with two contractions, set
\[
R_X(a,b)
=
\sum_{c\in \{a,b\}}
\Res_{w=c}
\frac{v(w)}
     {(w-a)^2(w-b)^2}\,dw,
\]
where a repeated point is counted only once.
For $a\neq b$, a direct computation gives
$R_X(a,b)
=
\frac{v'(a)+v'(b)}{(a-b)^2}
-
2\frac{v(a)-v(b)}{(a-b)^3}
=
B_X(a,b)$. To include the diagonal, choose $0<R<1$ such that $|a|,|b|<R$. By the residue theorem,
$R_X(a,b)
=
\frac{1}{2\pi i}
\int_{|w|=R}
\frac{v(w)}{(w-a)^2(w-b)^2}\,dw,$
which is holomorphic on $B_R(0)\times B_R(0)$. Hence, $R_X=B_X$ also along the diagonal by the identity theorem.
Thus, $
\sum_{c\in S}
\Res_{w=c}
v(w)
\left(
\pa_a^m\frac1{(w-a)^2}
\right)
\left(
\pa_b^\ell\frac1{(w-b)^2}
\right)\,dw
=
\pa_a^m\pa_b^\ell B_X(a,b).
$
Hence, the assertion follows from Proposition \ref{prop_app_global_explicit}.
\end{proof}

\begin{cor}\label{cor_TT_OPE}
For any $X=v(z)\pa_z\in\Vh$ and $w\in\bD$, we have
\begin{align}
\left.\frac{d}{dt}\right|_{t=0}
(\Psi_{t}^X)_*(T(w))
=
v(w)J^{(1)}(w)J(w)
+
2v'(w)T(w)
+
\frac{1}{12}v'''(w)\mathbf 1.
\label{eq_energy_infinitesimal}
\end{align}
\end{cor}
Corollary~\ref{cor_TT_OPE} follows from the conformal Ward identity in two ways.
The first is the residue computation familiar from vertex operator algebras.
In fact, \eqref{eq_app_Tproduct} gives 
\[
T(w_1)*T(w_2)
=
T(w_1)T(w_2)
+\frac{J(w_1)J(w_2)}{(w_1-w_2)^2}
+\frac{1}{2(w_1-w_2)^4},
\]
and the term $\frac{1}{12}v'''(w)$ in \eqref{eq_energy_infinitesimal} comes from $\frac{1}{2(w_1-w_2)^4}$, which is the central-charge term in the operator product expansion of $T(w_1)T(w_2)$.
The second is obtained directly from the flow. 
By Proposition \ref{prop_app_global_explicit},
\begin{align*}
\left.\frac{d}{dt}\right|_{t=0}(\Psi_{t}^X)_*T(w)
=
\left.\frac{d}{dt}\right|_{t=0}
\left(
(\Psi_{t}^X{}'(w))^2T(\Psi_{t}^X(w))
+
\left.
\pa_{z_1}\pa_{z_2}
H_{\Psi_{t}^X}(z_1,z_2)
\right|_{z_1=z_2=w}
\right)
\end{align*}
and by \eqref{eq_app_def_Bv} and \eqref{eq_app_KB}
\begin{align}
\left.\frac{d}{dt}\right|_{t=0}
\left.
\pa_{z_1}\pa_{z_2}
H_{\Psi_{t}^X}(z_1,z_2)
\right|_{z_1=z_2=w}
=\frac{1}{2}B_X(w,w)= \frac{1}{12}v'''(w).
\label{eq_H_infinitesimal_diagonal}
\end{align}
Thus, the term involving $v'''(w)$, which underlies the
Gelfand--Fuchs cocycle, is the infinitesimal contribution
of the harmonic correction $H_\phi$ arising from the change of
boundary conditions of the Green function (see  Remark \ref{rem_GF}).

%
%

\begin{rem}
\label{rem_GF}
Let $\ODm$ denote the space of meromorphic functions on $\bD$ with at most a pole at $0$. Denote the space of meromorphic vector fields and quadratic differentials by
\begin{align*}
\Vm&=\ODm\pa_z \qquad \Om^2(\bD,*0)=\ODm dz^2.
\end{align*}
The Gelfand--Fuchs $2$-cocycle on $\Vm$ \cite{GF} is given by
\begin{align*}
\omega_{\mathrm{GF}}
\bigl(f(z)\partial_z,g(z)\partial_z\bigr)
=
\frac{1}{12}
\Res_{z=0}
f'''(z)g(z)\,dz,
\end{align*}
and, under the residue pairing, it determines the $1$-cocycle valued in $\Om^2(\bD,*0)$ by
\begin{align}
c_{\omega_{\mathrm{GF}}}:
\Vm\rightarrow\Om^2(\bD,*0),\quad
f(z)\pa_z\mapsto
\frac{1}{12}f'''(z)\,dz^2.
\label{eq_GF_one_cocycle}
\end{align}
Hence, \eqref{eq_H_infinitesimal_diagonal} identifies the infinitesimal contribution of the
harmonic cocycle $H_\phi$ with the restriction of the Gelfand--Fuchs $1$-cocycle $c_{\omega_{\rm GF}}$ to holomorphic vector fields.
\end{rem}

\appendix
\section{Remarks on distributions}\label{sec_app}
In the Appendix, for the reader's convenience, we review, following \cite{Hormander,RS}, the topologies on $C^\infty(U)$ and 
the definition of distributions $\Cc(U)'$.
The contents of this section are entirely standard and contain nothing new.

For any subset $S \subset \R^d$, set 
\begin{align*}
\Cc(S) = \{f \in \Cc(\R^d) \mid \mathrm{supp}f \subset S  \}.
\end{align*}
Let $U \subset \R^d$ be an open subset.
\begin{dfn}\cite[Definition 2.1.1]{Hormander}
\label{def_distribution}
A distribution $T$ in $U$ is a linear map $T: \Cc(U) \rightarrow \R$ such that for every compact set $K \subset U$, there exist constants $C$ and $k$ such that
\begin{align*}
|T(\phi)|\leq C \sum_{|\al|\leq k} \sup |\pa^\al \phi| 
\end{align*}
for any $\phi \in C_c^\infty(K)$. 
\end{dfn}
The vector space of distributions in $U$ is denoted by $\Cc(U)'$.
\begin{dfn}\cite[Definition 2.2.2]{Hormander}\label{def_support}
Let $T \in \Cc(U)'$. For an open subset $V \subset U$, we say that $T$ vanishes on $V$ if $T(f)=0$ holds for every $f\in \Cc(V)$. The support of $T$ is defined as the complement of the largest open subset on which $T$ vanishes.
\end{dfn}

\begin{dfn}[$C^\infty$ topology of local uniform convergence]\label{def_topology}
Let $(K_n )_{n\in \N}$ be an increasing sequence of compact subsets such that $K_n \subset K_{n+1}^\circ$ and $\cup_n K_n=U$.
The topology on $C^\infty(U)$ is defined by the family of seminorms
\begin{align*}
p_{K_n,m}(f) = \max_{|\al|\leq m} \sup_{x\in K_n}|\partial^\al f(x)|
\end{align*}
This does not depend on the particular choice of $(K_n)_n$.
\end{dfn}

Denote by $C^\infty(U)'$ the space of all continuous linear maps $T:C^\infty(U) \rightarrow \R$, the topological dual vector space.
Note that a linear map $T:C^\infty(U) \rightarrow \R$ is continuous if and only if
there is a compact subset $K \subset U$ and $C>0$, $m \geq 0$ such that
\begin{align}
|T(f)|\leq C p_{K,m}(f)\label{eq_app_K}
\end{align}
for any $f \in C^\infty(U)$.

\begin{prop}{\cite[Theorem 2.3.1]{Hormander}}\label{app_dist_compact}
$T \in \Cc(U)'$ is a restriction of an element in $C^\infty(U)'$ if and only if $\supp(T)$ is compact.
\end{prop}
\begin{proof}
Let $T \in C^\infty(U)'$ satisfy \eqref{eq_app_K}. Then, the restriction $T|_{\Cc(U)}:\Cc(U) \rightarrow \R$ is a distribution with $\supp(T) \subset K$.
Conversely, let $T \in \Cc(U)'$ be a distribution with compact support in $U$.
Let $\eta \in \Cc(U)$ be a function such that 
$\eta \equiv 1$ on a neighborhood of $\supp(T)$.
Define a linear map $\tilde{T}:C^\infty(U) \rightarrow \R$ by
\[
\tilde{T}(f) = \langle T,\eta f \rangle
\qquad (f \in C^\infty(U)).
\]
The definition of $\tilde{T}$ is independent of the choice of $\eta$.
Moreover, since \eqref{eq_app_K} holds with $K=\supp(\eta)$, the map $\tilde{T}$ is continuous.
Hence $\tilde{T} \in C^\infty(U)'$.
\end{proof}

\begin{rem}
It is important to note that \eqref{eq_app_K} does not in general hold with $K=\supp(T)$.
In fact, in general, one cannot choose $\eta$ with $\supp(\eta)=\supp(T)$ (For a counterexample, see \cite[Example 2.3.2]{Hormander}).
\end{rem}
Thus, elements of $C^\infty(U)'$ may be identified with compactly supported distributions.
%

%

\begin{rem}
The space of distributions $\Cc(U)'$ can also be regarded as the topological dual 
by endowing $\Cc(U)$ with a suitable locally convex topology
(see \cite[Section V]{RS}).
\end{rem}

Let $u \in \Cc(\R^d)$ and $T \in \Cc(\R^d)'$. The convolution $T*u$ is a map $\R^d \rightarrow \R$ defined by
\begin{align*}
(T*u)(x) = T(u(x-\bullet)).
\end{align*}
\begin{prop}\cite[Theorem 4.1.1]{Hormander}\label{prop_convolutions}
Let $u \in \Cc(\R^d)$ and $T \in \Cc(\R^d)'$. Then, $T*u \in C^\infty(\R^d)$ with
\begin{align*}
\supp(T*u) \subset \supp(u) + \supp(T)
\end{align*}
and for any $\al \in \N^d$, $\partial^\al (T*u) = (\partial^\al T)*u =T*(\partial^\al u)$.
\end{prop}

\begin{dfn}\cite[Definition 4.2.2]{Hormander}
\label{def_convolution}
The convolution of $T_1,T_2 \in C^\infty(\R^d)'$ is defined to be the unique distribution $T$ such that
\begin{align*}
T*f = T _1* (T_2* f),
\end{align*}
for any $f \in \Cc(\R^d)$. 
\end{dfn}

\begin{prop}\cite[Theorem 4.2.4]{Hormander}
\label{prop_conv_distribution}
For any $T_1,T_2 \in C^\infty(\R^d)'$,
\begin{align*}
\supp(T_1*T_2) \subset \supp(T_1)+\supp(T_2)
\end{align*}
and for any $\al \in \N^d$, $\partial^\al (T_1*T_2) = (\partial^\al T_1)*T_2 =T_1*(\partial^\al T_2)$.
Moreover, $\delta_0 *T =T = T*\delta_0$ for any $T \in C^\infty(\R^d)'$.
\end{prop}

Set
\begin{align}
\begin{split}
\Delta_x &=- (\partial_{x_1}^2+\dots+\partial_{x_d}^2),\\
\Psi(x) &= 
\begin{cases}
-(2\pi)^{-1}\log|x| & d=2,\\
\frac{1}{(d-2)\om_d\norm{x}^{d-2}}& d \geq 3,
\end{cases}
\end{split}
\label{eq_Laplace_fund}
\end{align}
where $\om_d$ is a volume of $S^{d-1}$.
Then, $\Psi(x)$ is locally integrable, and defines a distribution in $\Cc(\R^d)'$ given by $\langle\Psi,f \rangle = \int_{\R^d} \Psi(x)f(x) d^dx$
for any $f \in \Cc(\R^d)$. It is well-known that
\begin{align*}
\Delta_x \Psi(x)=\delta_0
\end{align*}
as distributions \cite[Theorem 3.3.2]{Hormander}.
It is clear that 
\begin{align}
\log|z-w| = \log|z|+ \log|1-w/z| = \log|z|-\ft \sum_{n \geq 1} \frac{1}{n}((w/z)^n+(\bar{w}/\z)^n),
\label{eq_Harm_expansion2}
\end{align}
and the sum is absolutely convergent and uniformly convergent in any compact subset $K \subset \{(z,w) \in \C^2 \mid |z|>|w|\}$.
Similarly, for $d\geq 3$, let $(C_n(t))_{n \geq 0}$ be polynomials of $t$ given as the coefficients of the Taylor expansion
\begin{align*}
\frac{1}{(1-2tz+z^2)^{(d-2)/2}} = \sum_{n\geq 0} C_n(t)z^n,
\end{align*}
which is called the \textbf{Gegenbauer polynomials}.
Then, 
\begin{align}
\begin{split}
\frac{1}{\norm{x-y}^{d-2}} &= \norm{x}^{-(d-2)}\frac{1}{(1-2\frac{(x,y)}{\norm{x}\norm{y}}\norm{y}/\norm{x}+(\norm{y}/\norm{x})^2)^{\frac{d-2}{2}}}\\
&=\norm{x}^{-(d-2)}\sum_{n \geq 0}C_n\left(\frac{(x,y)}{\norm{x}\norm{y}}\right) (\norm{y}/\norm{x})^n\\
&=\sum_{n \geq 0}\norm{x}^{-(d-2)-2n}C_n\left(\frac{(x,y)}{\norm{x}\norm{y}}\right)\norm{y}^n \norm{x}^n,
\end{split}
\label{eq_Harm_expansion}
\end{align}
where the sum is uniformly convergent in any compact subset $K \subset \{(x,y)\in \R^{2d} \mid |x|>|y|\}$.
Moreover, let $\{Y_{n,k}(x)\}_{k=1,\dots,\dim \Harm_{d,n}}$ be an orthonormal basis of $\Harm_{d,n}$ with respect to the inner product on $L^2(S^{d-1})$ (see Section \ref{sec_expansion} and \eqref{eq_def_L2}).
By \cite[Theorem 1.2.6]{DX}, we have:
\begin{align*}
C_n\left(\frac{(x,y)}{\norm{x}\norm{y}}\right)\norm{y}^n \norm{x}^n = c_{d,n}\sum_{k=1}^{\dim \Harm_{d,n}}Y_{n,k}(x)Y_{n,k}(y),
\end{align*}
where $c_{d,n}=\frac{d-2}{2n+d-2}$.

\vspace{10mm}

\textbf{Acknowledgements}
I would like to thank Masahito Yamazaki for valuable discussions on factorization algebras, Atsushi Katsuda and Shota Hamanaka for discussions on the conformal Laplacian, and Masashi Kawahira and Tomohiro Shigemura for helpful discussions. I am also grateful to the anonymous referee for constructive comments that helped improve the paper.
This work is supported by Grant-in-Aid for Early-Career Scientists (24K16911).

%
%
%

\end{document}